\documentclass[12pt,reqno]{amsart}
\usepackage[utf8]{inputenc}
\usepackage[english]{babel}
\usepackage[T1]{fontenc}
\usepackage{amsmath}
\usepackage{amsfonts}
\usepackage{amssymb}
\usepackage{amsthm}
\usepackage{latexsym}
\usepackage{bbm}
\usepackage{graphicx}
\usepackage{color}
\usepackage[normalem]{ulem}

\usepackage{hyperref}
\hypersetup{citecolor=blue}
\hypersetup{colorlinks=true}
\hypersetup{linkcolor=blue}
\hypersetup{urlcolor=blue}
\newcommand*{\mmax}[1]{\mu_{\max}\left(#1\right)}
\newcommand*{\lmax}[1]{\lambda_{\max}\left(#1\right)}

\DeclareMathOperator{\obs}{\mathcal{O}}
\DeclareMathOperator{\state}{\mathcal{S}}
\DeclareMathOperator{\poly}{\mathcal{P}}
\DeclareMathOperator{\effect}{\mathcal{E}}
\DeclareMathOperator{\oA}{\mathsf{A}}
\DeclareMathOperator{\oB}{\mathsf{B}}
\newcommand{\aff}{\mathrm{aff}}
\newcommand{\ext}{\mathrm{ext}}
\newcommand{\conv}{\mathrm{conv}}

\newcommand{\real}{\mathbb R} %real
\newcommand{\complex}{\mathbb C} %complex
\newcommand{\nat}{\mathbb N} %natural
\newcommand{\hi}{\mathcal{H}} %Hilbert space H
\newcommand{\lh}{\mathcal{L(H)}} %bounded linear operators
\newcommand{\lsh}{\mathcal{L}_s(\mathcal{H})} %selfadjoint operators
\newcommand{\eh}{\mathcal{E(H)}} %effects
\newcommand{\ket}[1]{|#1\rangle} %ket
\newcommand{\kb}[2]{|#1\rangle\!\langle#2|} %ketbra
\newcommand{\no}[1]{\left\|#1\right\|} %norm
\newcommand{\tr}[1]{\mathrm{tr}\left[#1\right]} %trace
\newcommand{\id}{\mathbbm{1}} %identity operator
\newcommand{\A}{\mathsf{A}}%generic observable
\newcommand{\B}{\mathsf{B}}%generic observable
\newcommand{\sX}{\mathcal{X}} %subset of states
\newcommand{\sY}{\mathcal{Y}} %subset of states
\newcommand{\sZ}{\mathcal{Z}} %subset of states

\newcommand{\isymbol}{\mathfrak{I}} %choose symbol here
\newcommand{\is}[1]{\isymbol(#1)}
\newcommand{\isn}[1]{\isymbol_n(#1)}
\newcommand{\ism}[1]{\isymbol_m(#1)}
\newcommand{\nsat}{\mathfrak{n}}
\newcommand{\mdet}{\mathfrak{m}}
\newcommand{\dop}{\mathfrak{d}}
\newcommand{\expec}{\mathbb{E}}

\newtheorem{lemma}{Lemma}
\newtheorem*{lemma*}{Lemma}
\newtheorem{proposition}{Proposition}

\theoremstyle{definition}
\newtheorem{definition}{Definition}
\theoremstyle{remark}
\newtheorem{example}{Example}
\title[]{Encoding and decoding of information in general probabilistic theories}
\author{Teiko Heinosaari}
\email{teiko.heinosaari@jyu.fi}
\address{Faculty of Information Technology, University of Jyväskylä, Finland}

\author{Leevi Lepp\"aj\"arvi}
\email{leevi.leppajarvi@savba.sk}
\address{RCQI, Institute of Physics, Slovak Academy of Sciences,D\'ubravsk\'a cesta 9, 84511 Bratislava, Slovakia}

\author{Martin Pl\'{a}vala}
\email{martin.plavala@uni-siegen.de}
\address{Naturwissenschaftlich-Technische Fakult\"{a}t, Universit\"{a}t Siegen, Walter-Flex-Strasse 3, 57068 Siegen, Germany}
\date{}

\begin{document}

\maketitle

\begin{abstract}
    Encoding and decoding are the two key steps in information processing. In this work we study the encoding and decoding capabilities of operational theories in the context of information-storability game, where the task is to freely choose a set of states from which one state is chosen at random and by measuring the state it must be identified; a correct guess results in as many utiles as the number of states in the chosen set and an incorrect guess means a penalty of a fixed number of utiles. We connect the optimal winning strategy of the game to the amount of information that can be stored in a given theory, called the information storability of the theory, and show that one must use so-called nondegradable sets of states and nondegradable measurements whose encoding and decoding properties cannot be reduced. We demonstrate that there are theories where the perfect discrimination strategy is not the optimal one so that the introduced game can be used as an operational test for super information storability. We further develop the concept of information storability by giving new useful conditions for calculating it in specific theories.
\end{abstract}

%%%%%%%%%%%%%%%%%%%%%%%
\section{Introduction}
%%%%%%%%%%%%%%%%%%%%%%%

Let's have a birthday party! For entertainment and fun, we could choose to play a specific kind of game.
It goes like this: given a state space $\state$ you are allowed to pick $n \in \nat$ and a subset of states $\sX = \{s_i\}_{i=1}^n$. 
You will then be given a random unknown state from the set $\sX$ and allowed to perform a measurement. 
After obtaining a measurement outcome, you have to declare the label of the state you think you were given. If you are correct, then you win $n$ utiles\footnote{Utile is a theoretical unit of utility. For small amounts it can be replaced by money, but unlike to money, the law of diminishing returns does not apply to utiles.}, 
 while if you are wrong you get $w$ utiles. We will mostly be concerned with the cases when $w < 0$, meaning that you loose utiles if your guess is wrong.
 Here by a state space we mean any set of preparations closed with respect to randomization, this clearly includes both classical and quantum theories and such state spaces were treated by Holevo in his celebrated textbooks \cite{PSAQT82,SSQT01,QSCI12}.
 %, and also by many other researchers \cite{XXX}.
 The aforementioned game has a clear practical interpretation: the aim is to pick the largest possible set of states $\sX = \{s_i\}_{i=1}^n$ that can be efficiently used for storing classical information. Thus we will refer to this game as the \emph{information-storability game}. The information storability game is illustrated in Fig. \ref{fig:inf-st-game}.
 %(\blue{or information-storing game?})

 \begin{figure}
\begin{center}
    \includegraphics[width=0.8\textwidth]{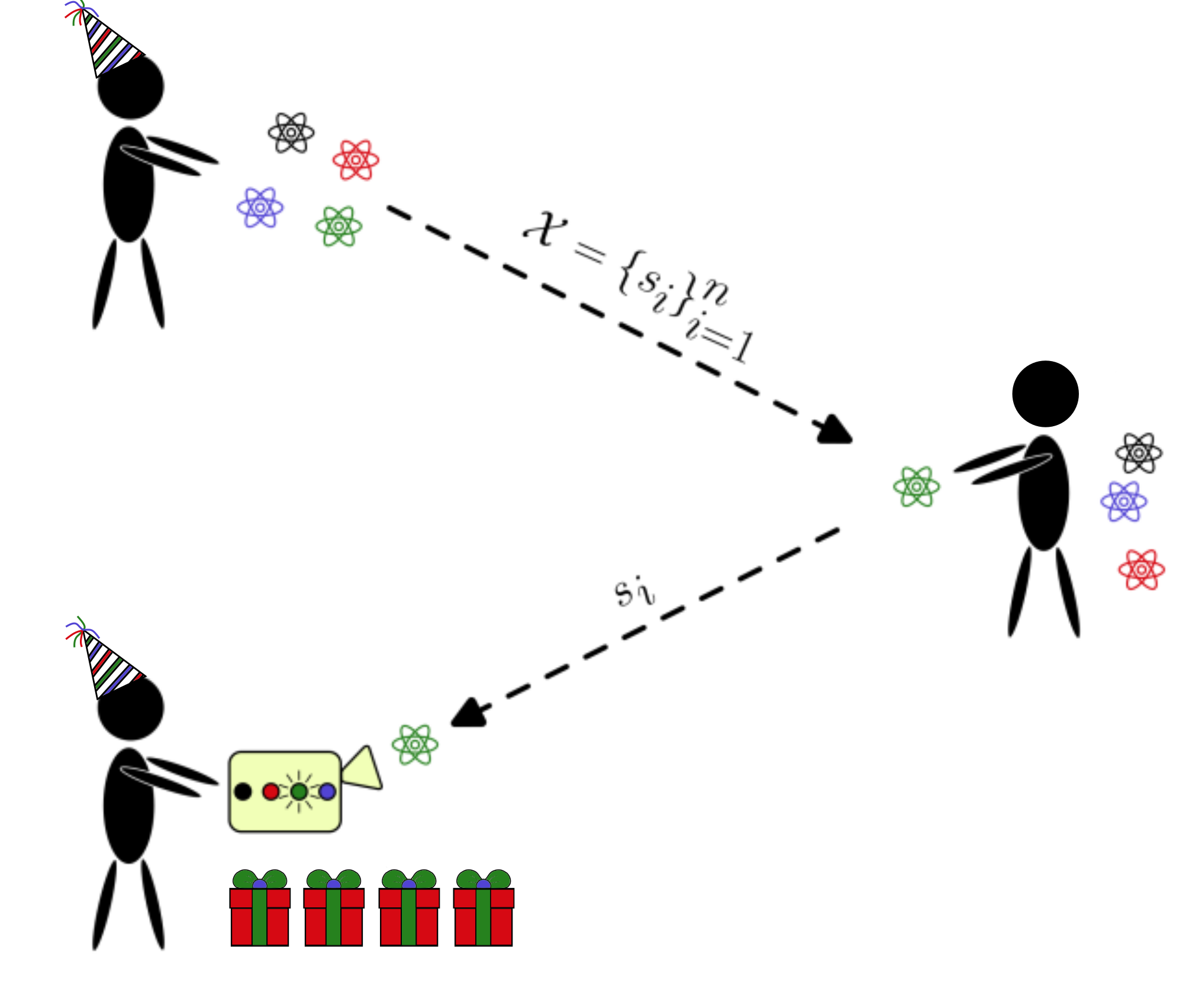}
\caption{ The party game of the year (the information storability game): you are free to choose a set $\sX$ of $n$ states from which a single state will be chosen and your task is to guess which state it was. Correct answer rewards you $n$ utiles (represented by presents in the figure) and incorrect answer costs you a penalty (some fixed number of utiles).}
\label{fig:inf-st-game}
\end{center}
\end{figure}

What is the most optimal strategy for the information-storability game? 
Clearly, one can choose $n=1$ in which case one always wins, but the payoff is small. 
A better strategy would be to pick $\sX$ to be one of the largest possible sets of perfectly distinguishable states, in which case a rational player still always wins and the payoff is given by the maximal number of perfectly distinguishable states in a given theory.
For instance, in the case of quantum theory the payoff would coincide with the dimension of the underlying Hilbert space. In general in quantum theory, the quantity we are working with is closely related to accessible information and thus upper bounded by the Holevo information \cite{Holevo73b}.
It is also clear that somehow artificially enlarging $\sX$ is not helpful since then the probability of guessing correctly decreases and that affects the expectation value of the payoff.
The questions we want to answer are the following. For a fixed state space $\state$ and a given penalty $w$, what is the best strategy in the information-storability game? 
Are there state spaces $\state$ where the optimal strategy also involves the possibility of guessing incorrectly?

Information storing depends both on encoding and decoding, hence one can look these and other related questions both from the perspectives of states and measurements.
The information storing capability of an abstract state space has been introduced and studied in \cite{MaKi18}. 
One of the main results of that investigation is that a state space with large information storing capability is necessarily highly asymmetric. 
This explains why higher dimensional quantum state spaces do not have point-symmetry as the qubit state space has.
In addition to these results, information storability can be compared to the operational dimension of the theory. 
In quantum and classical theories the two numbers are the same, but in general they can be different.
In \cite{MaKi18} a sufficient condition for the agreement of the numbers was derived and it is relying on the existence of suitable symmetry transformations. 
The symmetry group of a state space is a useful and descriptive feature of the theory. 
However, the symmetry group has no direct operational manifestation.

In the current work we develop the concept of information storability further. 
We introduce a binary division for all finite subsets of states into degradable and nondegradable sets, depending on their information storing property.
Roughly speaking, a degradable set has a proper subset that already stores the same information.
Information storing property of subsets of states has a dual notion for measurements, namely, decoding power.
Aided with these concepts, we present an operational classification of state spaces that is based on their information storing and decoding properties.
In this classification, quantum theory and classical theory belong to the same class, while there are some other state spaces that behave differently.

%%%%%%%%%%%%%%%%%%%%%%%
\section{Encoding and decoding of information} \label{sec:encoding-decoding}
%%%%%%%%%%%%%%%%%%%%%%%

In the information storing game the finite set of states (and thus also the number of states) can be chosen freely and the goal is to maximize the number of utiles that one gets by correctly distinguishing the states and by avoiding incorrect guesses which causes one to lose utiles. We will consider the information storing game in the framework of operational theories.

\subsection{Short introduction to operational theories}

The framework of operational theories is built on describing physical experiments by preparing, transforming, measuring and joining (physical) systems. In particular, we will use the convex formulation of \emph{general probabilistic theories (GPTs)}, where the set of states $\state$, i.e., all possible different preparation procedures of the system, is assumed to be a compact convex subset of a finite-dimensional real vector space (see e.g. \cite{LamiThesis,Plavala23,LeppajarviThesis} for more thorough reviews). There convexity is a natural requirement that arises from the possibility of probabilistic mixing of different preparation procedures so that if we have two preparation devices preparing states $s_1,s_2 \in \state$ we can assign a probability $p \in [0,1]$ and use the state $s_1$ with probability $p$ and state $s_2$ with probability $1-p$ in each round of an experiment leading to a mixed state $p s_1 + (1-p) s_2 \in \state$. On the other hand, a measurement $\oA$ with $n$ $(< \infty)$ outcomes is taken to be a collection of affine functionals $\oA_i: \state \to [0,1]$, called \emph{effects}, for all $i\in [n] := \{1, \ldots,n\}$ such that $\sum_{i=1}^n \oA_i(s) =1 $ for all states $s \in \state$. 
The interpretation then is that $\oA_i(s)$ describes the probability that an outcome $i$ is attained when we measure the system which is in state $s$ with a measurement $\oA$. Then the normalization $\sum_{i=1}^n \oA_i(s) =1 $, or equivalently, $\sum_{i=1}^n \oA_i =u $, where $u$ is the unit effect defined as $u(s)=1$ for all $s \in \state$, guarantees that some outcome is always attained in each measurement. We denote the set of effects on a state space $\state$, i.e., the set of affine functionals $e: \state \to [0,1]$ by $\effect(\state)$, and the set of measurements on $\state$ with $\obs(\state)$. The set of measurements with $n$ outcomes is then denoted by $\obs_n(\state)$.

 For our future analysis, following \cite{KiNuIm10} we say that a nonzero effect $e \in \effect(\state)$ is \emph{indecomposable} if any decomposition of $e$ into a sum of two other nonzero effects $e_1, e_2 \in \effect(\state)$ as $e = e_1 + e_2$ implies that $e= \lambda_1 e_1= \lambda_2 e_2$ for some $\lambda_1, \lambda_2 > 0$. It is known that indecomposable effects are exactly those that lie on the extreme rays of the positive dual cone $\{f \in \mathrm{span}(\state)^*\, | \, f(s) \geq 0 \ \forall s \in \state \}$ \cite{KiNuIm10}. We denote the set of indecomposable effects by $\effect_{ind}(\state)$ and the set of extreme indecomposable effects by $\effect^{ext}_{ind}(\state)$. Similarly we say that a measurement $\oA \in \obs(\state)$ is indecomposable if all of its nonzero effects are indecomposable and we denote the set of indecomposable measurement by $\obs_{ind}(\state)$.

\begin{example}[Quantum theory]
Consider a $d$-dimensional Hilbert space denoted as $\hi$. Let $\lh$ represent the algebra of linear operators acting on $\hi$, and $\lsh$ be the real vector space of self-adjoint operators on $\hi$. The state space $\state^q_d$ of a $d$-dimensional quantum theory consists of density matrices on $\hi$, defined as:
\begin{align*}
\state^q_d = \{ \varrho \in \lsh \, | \, \varrho \geq O, \ \tr{\varrho}=1 \},
\end{align*}
where $O$ is the zero-operator, and the partial order is induced by the cone of positive semi-definite matrices. The pure states, or extreme points of the state space, precisely correspond to rank-1 projections on $\hi$, i.e., operators of the form $\kb{\varphi}{\varphi}$ for some unit vector $\varphi \in \hi$.

The set of effects $\effect(\state^q_d)$ can be shown to be isomorphic to the set $\eh$ of self-adjoint operators bounded between $O$ and $\id$, where $\id$ is the identity operator on $\hi$, expressed as:
\begin{align*}
\effect(\state^q_d) \simeq \eh := \{ E \in \lsh \, | \, O \leq E \leq \id \}.
\end{align*}
Here, extreme effects align with projections on $\hi$. The indecomposable effects correspond to rank-1 operators of the form $\lambda \kb{\varphi}{\varphi}$ for some unit vector $\varphi \in \hi$ and $\lambda>0$ \cite{KiNuIm10}. Clearly effects of this form are extreme if and only if $\lambda=1$.

Measurements with a finite number of outcomes on $\hi$ are characterized by positive operator-valued measures (POVMs). These measures are represented as maps of the form $A: x \mapsto A(x)$, where $x$ belongs to some outcome set $[n]$, and the effect operators $A(x)$ in $\eh$ satisfy the condition $\sum_{x=1}^n A(x) = \id$. The indecomposable measurements correspond to POVMs with rank-1 effects.
\end{example}

\subsection{Degradable states and measurements in the information-storability game}

In the information-storability game, if one chooses a set of $n$ states $\sX=\{s_i\}_{i=1}^n \subset \state$ and a measurement $\oA \in \obs_n(\state)$ with $n$ outcomes, then given a random unknown state $s_i$ from $\sX$ the probability that one manages to distinguish that state from the others by using $\oA$ is just $\oA_i(s_i)$ in which case one gains $n$ utiles; otherwise one gains $w<0$ utiles, i.e., loses $|w|$ utiles with probability $1-\oA_i(s_i)$. Thus, the average number of utiles that one is rewarded when the game is repeated by using the states $\sX$ and the measurement $\oA$ is then
\begin{align}
 &   n \cdot \frac{1}{n}\sum_{i=1}^n \oA_i\left(s_i\right) + w \cdot \left[ 1- \frac{1}{n}\sum_{i=1}^n \oA_i\left(s_i\right) \right]  \nonumber \\
& = w + \left( 1- \frac{w}{n} \right) \sum_{i=1}^n \oA_i\left(s_i\right) =: \expec_w(\sX,\oA) \, , \label{eq:reward}
\end{align}
where we have assumed that in each round of the game the states are chosen with uniform probability. 

The optimal strategy for the game can be found by optimising both the set of states $\sX$ and the measurement $\oA$ in the above expression for $\expec_w(\sX,\oA)$. 
However, the order in which these two optimizations are performed does not matter  and thus we can look separately on both of these optimizations.

Let us start by first considering the states. The key element is choosing a set of states such that with high probability one manages to distinguish them. We define the encoding power of a set of states as the maximal success probability of the minimum error discrimination task of those states (which are chosen with equal apriori probability) multiplied by the  number of states.

 \begin{definition}
 For a finite set of states $\sX = \{s_i\}_{i=1}^n \subset \state$ we define the \emph{encoding power} $\mmax{\sX}$ of $\sX$ as
\begin{equation}
    \mmax{\sX}:= \sup_{\oA \in \obs_n(\state)} n \cdot \dfrac{1}{n}\sum_{i=1}^n \oA_i(s_i) = \sup_{\oA \in \obs_n(\state)} \sum_{i=1}^n \oA_i(s_i) \, .
\end{equation}
Since the set of measurements $\obs_n(\state)$ with fixed number of outcomes $n$ is convex and closed, the supremum is always attained by some measurement.
\end{definition}

Thus, for a given set $\sX$ of $n$ states the maximum average number of utiles that are awarded in the information-storability game by using the states in $\sX$ in terms of the encoding power $\mmax{\sX}$ is given by
\begin{equation}\label{eq:avg-reward-state}
    \mmax{\sX} + w \cdot \left(1- \dfrac{\mmax{\sX}}{n} \right) := \expec_w(\sX) \, .
\end{equation}
Then the maximum probability of gaining $n$ utiles by using the states $\sX$ is $\mmax{\sX}/n$. For example, if $\sX$ consists of $n$ perfectly distinguishable states, then $\mmax{\sX}= n$ and we are always guaranteed to get $n$ utiles.

When finding the optimizing set of states that maximizes the average number of rewarded utiles given by Eq. \eqref{eq:avg-reward-state} we note that we are not just trying to maximize $\mmax{\sX}$ but also the term  $\mmax{\sX}/n$.  In particular, if we find another set $\sY \subset \state$ with $\mmax{\sY} = \mmax{\sX}$ but which has $m < n$ states, then clearly the average number of rewarded utiles increases. In particular, given a set of states $\sX$, we can start looking for the more optimal set $\sY$ by looking at the subsets of $\sX$. Namely, it is straightforward to check that if $\sY \subset \sX$, then $\mmax{\sY} \leq \mmax{\sX}$ because only in the worst case the optimizing measurement for $\mmax{\sY}$ (equipped with suitable number of additional zero effects) will be optimal for $\mmax{\sX}$. However, it may happen that $\mmax{\sY} = \mmax{\sX}$ even for some proper subsets $\sY$ of $\sX$.

\begin{definition}\label{def:degradable-states}
A finite set $\sX\subset\state$ is \emph{degradable} if there exists a proper subset $\sY \subset \sX$ such that 
\begin{equation}
\mmax{\sY} = \mmax{\sX} \, .
\end{equation}
Otherwise $\sX$ is \emph{nondegradable}. 
\end{definition}

Clearly now in the information-storability game, the set of states for the optimal strategy has to be nondegradable.

\begin{example}[Degradable pure qubit states]\label{ex:qubit-degradable}
   In the qubit case let us take two non-orthogonal unit vectors $\varphi, \psi \in \complex^2$ and the unit vectors orthogonal to them $\varphi^ \perp, \psi^ \perp \in \complex^2$ and let us consider the following set of four pure qubit states 
   $$ 
   \sX= \{\kb{\varphi}{\varphi},\kb{\varphi^\perp}{\varphi^\perp},\kb{\psi}{\psi},\kb{\psi^\perp}{\psi^\perp}\} \subset \state^q_2 \, .
   $$
   We note that from the basic decoding theorem (see e.g. \cite{QPSI10}) it follows that the maximum encoding power of any set of states in qubit is 2 (see Sec. \ref{sec:is}). Clearly, since both $\sY:= \{\kb{\varphi}{\varphi},\kb{\varphi^\perp}{\varphi^\perp}\}$  and $\sZ:=\{\kb{\psi}{\psi},\kb{\psi^\perp}{\psi^\perp}\}\}$ contain two orthogonal pure states, we have that $\mmax{\sY}= \mmax{\sZ} = 2$ so that also $\mmax{\sX}=2$ and thus $\sX$ is degradable. Thus, in this case one of the optimal strategies to decode the information encoded in $\sX$ is just to discriminate the two states in one of the orthogonal pairs contained in $\sX$.
\end{example}

Dually, one can also approach the information-storability game by not fixating on the states but by considering the distinguishing measurements instead. In particular, we can define the decoding power of a measurement as the maximal success probability of the minimum error discrimination task of the optimal set of states which the measurement can distinguish (with equal apriori probabilities) multiplied by the number of outcomes of the measurement.

\begin{definition}
For a measurement $\oA \in \obs_n(\state)$ with $n$ outcomes the \emph{decoding power} $\lmax{\oA}$ of $\oA$ is defined as
\begin{equation}
    \lmax{\oA} := n \cdot \dfrac{1}{n} \sum_{i=1}^n \sup_{s \in \state} \oA_i(s) = \sum_{i=1}^n \sup_{s \in \state} \oA_i(s) \, .
\end{equation}
Since the state space $\state$ is convex and closed, the supremum is always attained.
\end{definition}

Thus, for a given measurement $\oA$ with $n$ outcomes, the maximal average number of utiles which one is rewarded with in the information-storability game by using the decoding measurement $\oA$ in terms of the decoding power $\lmax{\oA}$ reads as
\begin{equation}
     \lmax{\oA} + w \cdot\left(1- \dfrac{\lmax{\oA}}{n} \right) =: \expec_w(\oA) \, .
\end{equation}
Then the maximum probability of gaining $n$ utiles by using the measurement $\oA$ is $\lmax{\oA}/n$.

Analogously as in the case of states, when starting to find the optimizing measurement for the optimal strategy in the information-storing game, we can focus on measurements which in the above equation do not change the value of $\lambda_{\max}$ but which nevertheless have less number of outcomes. One natural way is to look at merging some of the outcomes of the measurement and see if $\lambda_{\max}$ reduces or not. Thus, we can define an analogous concept of degradability for measurements:

\begin{definition}\label{def:degradable-measurement}
A measurement $\oA$ with $n$ outcomes is \emph{degradable} if there exists a measurement $\oB$ with $m<n$ outcomes which can be obtained from $\oA$ by merging some of its effects such that
\begin{equation}
\lmax{\oA} = \lmax{\oB} \, .
\end{equation}
If $\oA$ is not degradable then it is \emph{nondegradable}. 
\end{definition}
We conclude that in the information-storing game, the decoding measurement for the optimal strategy has to be nondegradable.

In the context of the information-storability game the above definitions of degradability seem intuitive (at least to us). However, in general there are more general ways of obtaining states from other states and measurements from another measurements. For more general definitions of degradability (which we show to reduce to the ones presented above) can be found in Appendix \ref{appendix-degradability}. 

\subsection{Conditions on degradable states and measurements}
For a finer analysis of degradability it is useful to define also set dependent decoding power; we denote
\begin{equation}
\lmax{\oA |  \sX} = \sum_{i}\sup_{s \in \sX}\oA_i(s) 
\end{equation}
for all $\oA \in \obs_n(\state)$ and $\sX \subset \state$.
The following observation will be useful in later developments. See Appendix \ref{appendix-degradability} for the proof.

\begin{lemma}\label{lemma:mmax-lmax-subset}
For any finite subset of states $\sX \subset \state$, we have that 
\begin{equation}
\mmax{\sX} = \sup_{\oA} \lmax{\oA | \sX}
\end{equation}
and they both can maximized by the same measurement. 
\end{lemma}

In the following we look first conditions for a subset of states to be degradable and then a condition for a measurement to be degradable.

\subsubsection{Degradable sets of states}

Next we will give some sufficient conditions for a set of states to be degradable.

\newpage

\begin{proposition} \label{prop:X-degradable}
Let $\sX\subset\state$ be a finite set.
In the following cases $\sX$ is degradable:
\begin{itemize}
\item[(a)] $\sX$ has more than $\dim(\aff(\state))+1$ elements.
\item[(b)] $\sX \neq \ext(\conv(\sX))$
\end{itemize}
\end{proposition}
\begin{proof}
Let $\sX = \{s_1, \ldots, s_n\} \subset \state$ for some $n \in \nat$ and let $\oA$ be some $n$-outcome measurement that maximizes $\mmax{\sX}$. We see that by Lemma \ref{lemma:mmax-lmax-subset} we have that $\mmax{\sX} = \sup_{\oB} \lmax{\oB | \sX} = \lmax{\oA | \sX}$. By using the simulation scheme of measurements introduced in \cite{oszmaniec2017simulating,guerini2017operational,FiHeLe18} now $\oA$ can be written as a simulation of some finite number $M$ of extreme indecomposable measurements (also called as simulation irreducible measurements) $\{\oB^{(k)}\}_{k=1}^M$ with some finite number $m$ of outcomes with some probability distribution $(p_k)_{k=1}^M$ and row-stochastic postprocessing matrices $\nu^{(k)}=(\nu^{(k)}_{ji})_{j \in [m], i \in [n]}$ (so that $\nu_{ji} \geq 0$ for all $i\in[n]$ and $j \in [m]$ and and $\sum_i \nu_{ji} =1$ for all $j\in [m]$) such that
\begin{align}\label{eq:simulation}
    \oA_i = \sum_{k=1}^M \sum_{j=1}^{m} p_k \nu^{(k)}_{ji}\oB^{(k)}_j
\end{align}
for all $i \in [n]$. 

(a) We now observe that
\begin{align*}
    \lmax{\oA | \sX} &= \sum_{i=1}^n \sup_{l} \oA_i(s_l) = \sum_{i=1}^n \sup_l \left( \sum_{k=1}^M \sum_{j=1}^{m} p_k \nu^{(k)}_{ji}\oB^{(k)}_j\right)(s_l) \\
    &\leq \sum_{i=1}^n \sum_{k=1}^M \sum_{j=1}^{m} p_k \nu^{(k)}_{ji} \sup_l \oB^{(k)}_j(s_l) = \sum_{k=1}^M \sum_{j=1}^{m} p_k  \sup_l \oB^{(k)}_j(s_l) \\
    &=  \sum_{k=1}^M p_k \lmax{\oB^{(k)} | \sX} \leq \max_{k} \lmax{\oB^{(k)} | \sX}.
\end{align*}
Thus, since $\oA$ maximizes $\mmax{\sX}$ we get that also $\oB^{(k)}$ maximizes it for some $k \in [M]$. Since $\oB^{(k)}$ is an extreme measurement its nonzero effects must be linearly independent \cite{FiHeLe18} so that it has maximally $\dim(\aff(\state)) +1$ nonzero outcomes. Thus, if $n>\dim(\aff(\state)) +1$, then there exists a proper subset $\sY \subset \sX$ of at most  $\dim(\aff(\state)) +1$ states such that $\lmax{\oB^{(k)} | \sY } = \lmax{\oB^{(k)} | \sX } = \mmax{\sX}$ implying that $\mmax{\sY} = \mmax{\sX}$ so that $\sX$ is degradable.

(b) From Eq. \eqref{eq:simulation} we see that
\begin{align*}
    \mmax{\sX} &= \sum_{i=1}^n \oA_i(s_i) = \sum_{i=1}^n \sum_{k=1}^M \sum_{j=1}^{m} p_k \nu^{(k)}_{ji}\oB^{(k)}_j(s_i) \\
    &= \sum_{k=1}^M \sum_{j=1}^{m} p_k \oB^{(k)}_j\left(\sum_{i=1}^n\nu^{(k)}_{ji}s_i\right),
\end{align*}
where now $s^{(k)}_j := \sum_{i=1}^n\nu^{(k)}_{ji}s_i \in \conv(\sX)$ for all $j \in [m]$ for all $k \in [M]$. Then it follows that
\begin{align}
    \mmax{\sX} &= \sum_{k=1}^M \sum_{j=1}^{m} p_k \oB^{(k)}_j\left(s^{(k)}_j\right) \leq \sum_{k=1}^M p_k \lmax{\oB^{(k)} \left| \right. \conv(\sX)} \nonumber \\
    &\leq \max_k \lmax{\oB^{(k)} \left| \right. \conv(\sX)}. \label{eq:mmax-conv}
\end{align}
Since linear functions (such as effects) attain their supremums on the extreme points of a convex set, we have that $\lmax{\oB' |\conv(\sX)} = \lmax{\oB' | \ext(\conv(\sX))}$ for any measurement $\oB'$. Thus, if we denote $\sY = \ext(\conv(\sX))$ we have by Eq. \eqref{eq:mmax-conv} that there exists $k \in [M]$ such that $\mmax{\sX} \leq \lmax{\oB^{(k)} \left| \right. \sY} \leq \mmax{ \sY}$. If now $\sX \neq \ext(\conv(\sX))$, so that we have a strict inclusion $\sY \subset \sX$ and it follows that $\mmax{\sY} = \mmax{\sX}$ so that $\sX$ is degradable. 
\end{proof}

One can easily see that neither of the conditions is necessary for degradability in general: Let us consider the following set of three pure qubit states $\sX= \{\kb{\varphi}{\varphi}, \kb{\varphi^\perp}{\varphi^\perp}, \kb{\psi}{\psi}\}$ where $\{\ket{\varphi}, \ket{\varphi^\perp}\} \in \complex^2$ is some orthogonal basis in $\complex^2$ and $\ket{\psi}= \frac{1}{\sqrt{2}}(\ket{\varphi}+\ket{\varphi^\perp})$. Clearly $\sX$ is degradable (similarly as in Example \ref{ex:qubit-degradable}) since we have that $\mmax{\sX} = \mmax{\{\kb{\varphi}{\varphi}, \kb{\varphi^\perp}{\varphi^\perp}\}} =2$ but on the other hand $\sX$ has less than $\dim(\aff(\state^q_2))+1=4$ states and clearly $\sX = \ext(\conv(\sX))$.

%\blue{Are all affinely independent degradable sets as in the previous qubit example? Does this hold: a set of states is nondegradable if and only if it is affinely independent and does not contain maximal number of distinguishable states?}

A potential generalization of Prop.~\ref{prop:X-degradable} would be to prove that $\sX$ is degradable if the states in $\sX$ are not affinely independent. This is not true in general, as demonstrated by the following example:
\begin{example}
Let $\state = S_4$ be the simplex generated by 4 linearly independent states $x_1, x_2, x_3, x_4$. There is only one extreme indecomposable (simulation irreducible) measurement $b$ given by effects $b_i \in \effect^{ext}_{ind}(\state)$, $i \in \{1, \ldots, 4\}$ such that $b_i(x_j) = \delta_{ij}$ for $i,j \in \{1, \ldots, 4\}$ where $\delta_{ij}$ is the Kronecker delta. Define the states
\begin{align}
x_{1123} &= \dfrac{x_1}{2} + \dfrac{x_2 + x_3}{4}, \\
x_{1224} &= \dfrac{x_2}{2} + \dfrac{x_1 + x_4}{4}, \\
x_{1334} &= \dfrac{x_3}{2} + \dfrac{x_1 + x_4}{4}, \\
x_{2344} &= \dfrac{x_4}{2} + \dfrac{x_2 + x_3}{4},
\end{align}
and let $\sX = \{ x_{1123}, x_{1224}, x_{1334}, x_{2344} \}$. We clearly have that $x_{1123} + x_{2344} = x_{1224} + x_{1334}$ so the states in $\sX$ are not affinely independent. Using the result of Lemma \ref{lemma:mmax-lmax-subset} that $\mmax{\sX} = \sup_{\oA} \lmax{\oA  |  \sX}$ and the fact that the previous supremum can be attained by extreme indecomposable measurements (see Appendix \ref{appendix-degradability} for proof), it follows that since there is only one simulation irreducible measurement on $S_4$ we must have
\begin{equation}
\mmax{\sX} = \sum_{i = 1}^4 \sup_{s \in \sX} b_i(s) = 2
\end{equation}
where we have used that $\sup_{s \in \sX} b_i(s) = \frac{1}{2}$, which is easy to check directly.

Now consider $\sY =  \{ x_{1123}, x_{1224}, x_{1334} \}$, then we get
\begin{equation}
\mmax{\sY} = \sum_{i = 1}^4 \sup_{s \in \sX} b_i(s) = \dfrac{3}{2} +  \sup_{s \in \sY} b_3(s) = \dfrac{7}{4} < 2,
\end{equation}
and analogically for the other subsets of $\sX$. Thus $\sX$ is nondegradable.
\end{example}

Conditions that are both necessary and sufficient for degradibility of sets of states remains an open problem.

\subsubsection{Degradable measurements} For measurements we can actually show a necessary and sufficient condition for degradability.

\begin{proposition}\label{prop:nondeg-measurement}
A measurement is nondegredable if and only if each of its effects attains its maximum value only on some different set of states.
\end{proposition}

\begin{proof}
    Let us take $\oA \in \obs_n(\state)$ and let us set $S_i = \{s \in \state \, | \, \oA_i(s) = \sup_{t \in \state} \oA_i(t) \}$ for all $i \in [n]$. Suppose that there are two indices $k,k' \in [n]$, $k \neq k'$ such that $S_k \cap S_{k'} \neq \emptyset$. Without loss of generality we may assume that $k =n-1$ and $k' = n$. Thus, we can take a set of $n$ states $ \{s_i\}_{i=1}^n$ such that $s_i \in S_i$ for all $i \in [n]$ and $s_{n-1} = s_{n}$. We can now define a measurement $\oB \in \obs_{n-1}(\state)$ by setting $\oB_i = \oA_i$ for all $i < n-1$ and $\oB_{n-1} = \oA_{n-1} + \oA_n$. We now see that 
    \begin{equation*}
        \lmax{\oA} = \sum_{i=1}^n \oA_i(s_i) = \sum_{i=1}^{n-1} \oB_i(s_i) \leq \lmax{\oB}.
    \end{equation*}
    However, since $\oB$ is formed by just merging some outcomes of $\oA$ we must have that $\lmax{\oB} \leq \lmax{\oA}$. Thus, $\lmax{\oA} = \lmax{\oB}$ so that $\oA$ is degradable.

    For the other direction, let now $\oA \in  \obs_n(\state)$ be a measurement such that each of its effects attains its maximum value only on some different set of states meaning that $S_i \cap S_j= \emptyset$ for all $i \neq j$. Suppose now that $\oA$ is degradable. Thus, there exists a measurement $\oB \in \obs_m(\state)$ with $m<n$ such that $\oB$ is formed by merging some of the effects of $\oA$ and $\lmax{\oA} = \lmax{\oB}$. In particular there exists a function $f: [n] \to [m]$ such that $\oB_j = \sum_{i \in f^{-1}(j)} \oA_i$ for all $j \in [m]$. Since $\lmax{\oA} = \lmax{\oB}$, there exists some set of $m$ states $ \{t_j\}_{j=1}^m$ such that $\sum_j \oB_j(t_j) = \lmax{\oB}$. Now we see that
    \begin{align*}
        \lmax{\oA} &= \lmax{\oB}= \sum_{j=1}^m \oB_j(t_j) = \sum_{j=1}^m \sum_{i \in f^{-1}(j)} \oA_i(t_j) = \sum_{i=1}^n \oA_i(t_{f(i)}) \\
        &\leq  \sum_{i=1}^n \sup_{s \in \state}\oA_i(s) = \lmax{\oA} .
    \end{align*}
    In particular this means that $t_{f(i)} \in S_i$ for all $i \in [n]$. However, since $S_i \cap S_j = \emptyset $ for all $i \neq j$ we must have that $f(i) \neq f(j)$ for all $i \neq j$ so that $f$ is injective. This is not possible since $m < n$. Hence, $\oA$ must be nondegradable.
\end{proof}

\begin{example}
    It is a straighforward consequence of Prop. \ref{prop:nondeg-measurement} that a POVM is nondegradable if and only if the set of eigenvectors corresponding to the maximal eigenvalue an effect operator is unique for every effect. Consequently, a rank-1 POVM is nondegradable if and only if none of its effects are proportional. It is also clear that a POVM that has only projections in its range is nondegradable.
\end{example}

%%%%%%%%%%%%%%%%%%%%%%%%%%%%%
\section{Information storability} \label{sec:is}
%%%%%%%%%%%%%%%%%%%%%%%%%%%%
Let us first focus on the case when $w=0$ in the information-storing game. Then one does not have to worry about the penalty and the task simplifies to finding a (finite) set of states that maximizes $\mu_{\max}$, or equivalently, a measurement which maximizes $\lambda_{\max}$. Thus, the maximum average number of utiles that the optimal strategy in this case rewards is
\begin{equation}
    \is{\state} := \sup_{\oA \in \obs(\state)} \lmax{\oA} = \sup_{\sX \subset \state} \mmax{\sX},
\end{equation}
which is known as the \emph{information storability} of the state space $\state$. In a given theory it corresponds to the maximal amount of information that can be successfully encoded and decoded in a communication scenario: namely, it is the maximisation of the product of the length of the message which one wishes to communicate and the probability of encoding/decoding the message. In other words, it is the maximal product of the number of states that is used to encode the message and the optimal success probability of decoding the message by performing a minimum-error discrimination task on the states.
We further denote
\begin{equation}
    \isn{\state} := \sup_{\oA \in \obs_n(\state)} \lmax{\oA} = \sup_{\sX \subset \state, |\sX|=n} \mmax{\sX},
\end{equation}
which is the information storability when the messages are limited to the length of $n$.
Clearly, $\isn{\state}\leq \ism{\state}$ whenever $n\leq m$, and we denote by $\nsat_{\state}$ the smallest integer $n$ such that $\isn{\state}=\is{\state}$.

We recall that the \emph{operational dimension} of a state space $\state$ is the largest number of perfectly distinguishable states contained in $\state$, and we denote it by $\dop_{\state}$. It is clear that $\nsat_{\state} \geq \dop_{\state}$. As noted in Example \ref{ex:qubit-degradable} in the case of the qubit, it is consequence of the basic decoding theorem (see e.g. \cite{QPSI10}) that in general in $d$-dimensional quantum theory we have that $\is{\state^q_d} = \dop_{\state^q_d}=\nsat_{\state^q_d}=d$ (see also Example \ref{ex:is-quantum} below). Curiously, there are state spaces $\state$ where $\nsat_{\state} > \dop_{\state}$ and thus $\is{\state} > \dop_{\state}$, in which case we say that the theory manifests \emph{super information storability} (see Example \ref{ex:is-polygon} below).

The concept of information storability was first introduced in \cite{MaKi18} and furthermore they showed that it is linked to the point-asymmetry of the state space. In particular, the information storability $\is{\state}$ of a state space $\state$ is related to the amount of asymmetry of $\state$ given by the (affine-invariant) Minkowski measure $\mathfrak{M}(\state)$ by the relation $\mathfrak{M}(\state) = \is{\state}-1$. In particular this means that $\is{\state}=2$ if and only if the state space $\state$ is point-symmetric (also called centrally-symmetric \cite{LaPaWi18,BlJeNe22}) meaning that there exists a state $s_0 \in \state$ such that for any state $s \in  \state$ there exists another state $s' \in \state$ such that $s_0 = \frac{1}{2}(s+s')$.
%\begin{theorem}(\cite{matsumoto2018information})
%In every state space $\state$ 
%\end{theorem}

Using the introduced notation we see that 
\begin{equation}\label{eq:bdt}
\mmax{\sX} \leq \is{\state} \, ,
\end{equation}
where the upper bound $\is{\state}$ does not depend on the size of the set $\sX$.
This motivates the following definition.

\begin{definition}
A finite set $\sX\subset\state$ that satisfies $\mmax{\sX} = \is{\state}$ is \emph{maximally decodable}.
\end{definition}

\begin{example}(\emph{Maximally decodable sets of pure quantum states}) 
Let us consider $d$-dimensional quantum state space.
Any set of $d$ orthogonal pure states is maximally decodable, but there are also other sets. Namely, let $P_i$, $i=1,\ldots,n$, be pure quantum states  such that
\begin{equation}\label{eq:Pi=r}
\sum_{i=1}^n P_i = r \id
\end{equation}
for some $r\in\real$. (The sum of projections is a positive operator, hence $r > 0$ if it exists.)
The orthogonal case is equivalent to \eqref{eq:Pi=r} with $r=1$, but $r$ can be different and this is hence more general condition.
For each $i$ we define an operator $\oA_i=\tfrac{1}{r} P_i$. 
These operators form a POVM and we have
\begin{equation}
\sum_i \tr{P_i \oA_i} = \frac{1}{r} \sum_i \tr{P^2_i} = \frac{1}{r} \sum_i \tr{P_i} = \frac{1}{r} \tr{\sum_i P_i} = d \, , 
\end{equation}
hence the set $\sX=\{P_1,\ldots,P_n\}$ is maximally decodable. 
A class of sets of the previous kind can be formed by choosing a $d$-dimensional irreducible unitary representation $g \mapsto U_g$ of a finite group $G$ and fixing a pure quantum state $P$.
If $U_g P U_g \neq P$ for all $g \neq e$, the pure states are labeled with the elements of $G$ as $P_g := U_g P U_g^*$. 
More generally, the pure states are labeled with the elements of the quotient group $G/H$, where $H=\{ h \in G : U_hPU_h^*=P\}$. Fig. \ref{fig:qubit-max-dec} shows an example of a maximally decodable (non-orthogonal) set of qubit states.

\begin{figure}
\begin{center}
    \includegraphics[width=7cm]{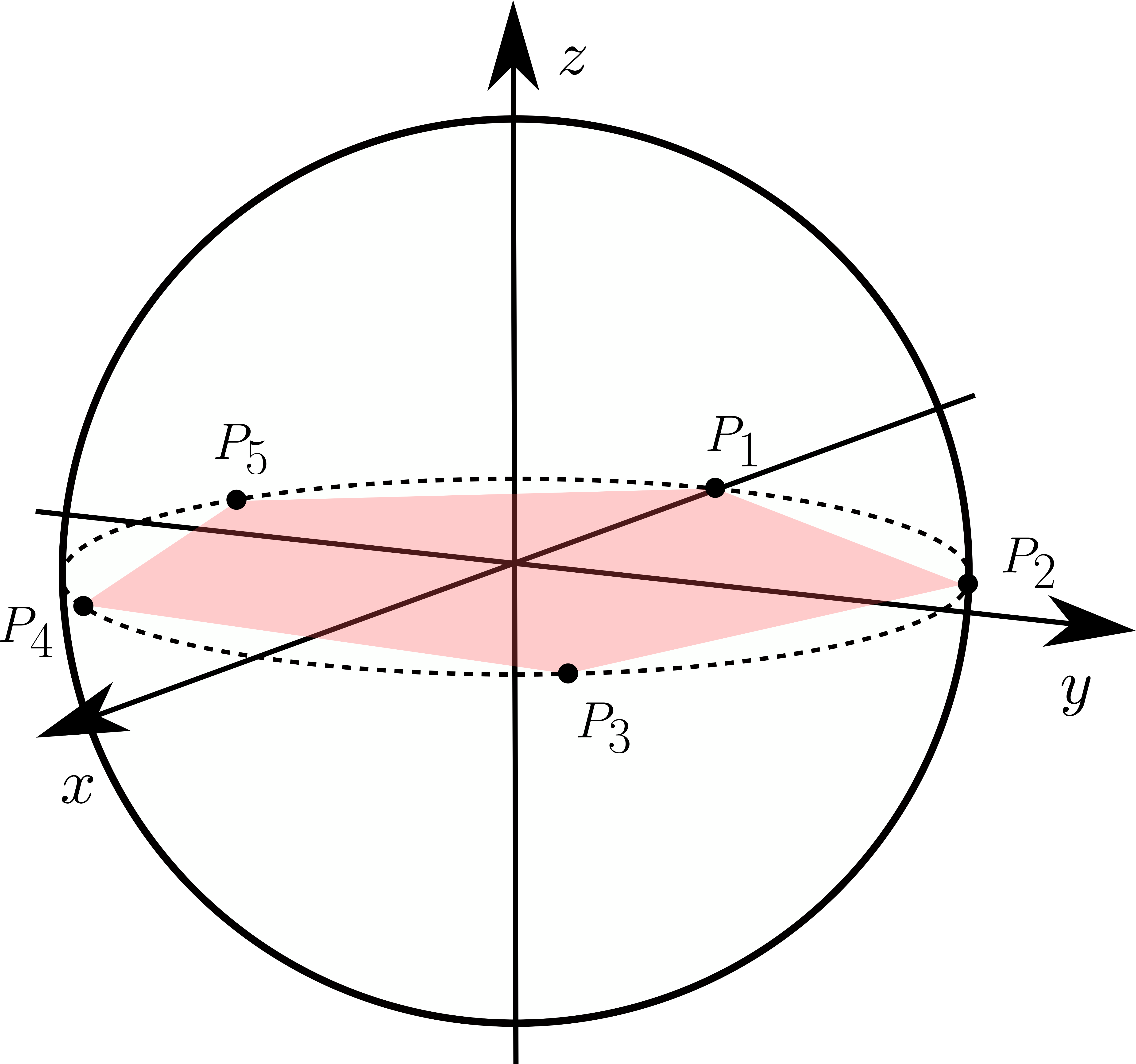}
\caption{Five maximally decodable pure qubit states on the Bloch sphere. }
\label{fig:qubit-max-dec}
\end{center}
\end{figure}

\end{example}

As a generalization of Example \ref{ex:qubit-degradable} we can show the following results connecting degradability and maximal decodability.

\begin{proposition}
Let $\sX\subset\state$ be a finite set. Suppose there exists a proper subset $\sY \subset \sX$ that is maximally decodable. Then $\sX$ is both degradable and maximally decodable.
\end{proposition}
\begin{proof}
    The claim follows straightforwardly by noting that since the maximally decodable set $\sY$ is a proper subset of $\sX \subset \state$ we must have that $\is{\state}= \mmax{\sY} \leq\mmax{\sX} \leq \is{\state}$.
\end{proof}

\begin{example}(\emph{Maximally decodable set need not be degradable})
It is clear that a maximal set of perfectly distinguishable states is maximally decodable but not degradable. One can then ask if a maximally decodable set with more than $d$ states has to be degradable. The answer is negative. One can see this for example by considering the trine states in qubit: let us fix some orthogonal basis $\{\ket{\varphi}, \ket{\varphi^\perp}\} \subset \complex^2$ and let us define
\begin{align*}
    \ket{\psi_1}&=\frac{1}{\sqrt{2}}\left( \ket{\varphi}+\ket{\varphi^\perp}\right) \\
    \ket{\psi_2}&=\frac{1}{\sqrt{2}}\left( \ket{\varphi}+e^{\frac{2\pi i}{3}}\ket{\varphi^\perp}\right) \\
    \ket{\psi_3}&=\frac{1}{\sqrt{2}}\left( \ket{\varphi}+e^{\frac{4\pi i}{3}}\ket{\varphi^\perp}\right) 
\end{align*}
and take $\sX = \{ \kb{\psi_j}{\psi_j}\}_{j=1}^3 \subset \state^q_2$. It is known that the optimal measurement for minimum error discrimination is just the POVM which measures along the directions of the state (see e.g. \cite{BaKuMoHi97}). In particular, if we set $A_j = \frac{2}{3}\kb{\psi_j}{\psi_j} $ it is straightforward to check that $A$ is a POVM and that $\sum_j A_j(s_j) =2= \is{\state^q_2}$ so that $\mmax{\sX}=2$. However, since none of the states in $\sX$ are orthogonal the set cannot be degradable.
\end{example}

Dually, we can also define maximal decodability for measurements.
\begin{definition}
A measurement $\oA$ that satisfies $\lmax{\oA} = \is{\state}$ has \emph{maximal decoding power}.
\end{definition}

Now we can give our conditions for calculating the information storability in specific theories and illustrate the use of them in some important examples.

\newpage

\begin{proposition}\label{prop:inf-stor-condition}
Let $\state$ be a state space.
\begin{enumerate}
    \item[i)] If there exists a state $s_0 \in \state$ such that $e(s_0) = f(s_0) =: \lambda_0$ for all $e,f \in \effect^{ext}_{ind}(\state)$, then $\is{\state} = 1/\lambda_0$ and it is attained by all indecomposable measurements.
    \item[ii)] Furthermore, if the maximal value of each $e \in \effect^{ext}_{ind}(\state)$ is attained on a unique pure state $s_e \in \state^{ext}$ such that $s_e \neq s_f$ for all $e \neq f$, $e,f \in \effect^{ext}_{ind}(\state)$, then $\is{\state}= 1/\lambda_0$ is attained only by indecomposable measurements.
\end{enumerate}
\end{proposition}
\begin{proof}
In \cite{HeLe22} it was shown that if the condition in \emph{i)} holds, then $\lmax{\state} = 1/\lambda_0$ and that $\lmax{\oA}= 1/ \lambda_0$ for all $\oA \in \obs_{ind}(\state)$. We will show that if also the condition in \emph{ii)} holds, then if $\lmax{\oA} = 1/\lambda_0$, then $\oA \in \obs_{ind}(\state)$.

First let us introduce the supremum norm $\no{\cdot}: \mathrm{span}(\effect(\state)) \to [0,1]$ on the set of effects defined by $\no{e} := \sup_{s \in \state} |e(s)| = \sup_{s \in \state} e(s)$ for all $e \in \effect(\state)$. It is clear that $0\leq \no{e} \leq 1$ for all effects $e \in \effect(\state)$ and if $e$ is an extreme effect then $\no{e} =1$. Also, in terms of the supermum norm we now have that $ \lmax{\oA} = \sum_{i=1}^n \no{\oA_i}$ for all $\oA \in \obs_n(\state)$.

 Let now $\oA$ be a measurement with $n$ outcomes. Each effect $\oA_i$ can be decomposed as $\oA_i = \sum_{j}^{r_i} \alpha^{(i)}_j e^{(i)}_j$ for some $r_i \in \nat$, some set of positive numbers $\{\alpha^{(i)}_j\}_{j=1}^{r_i} \subset \real_+$ and some distinct indecomposable extreme effects $\{e^{(i)}_j\}_{j=1}^{r_i} \subset \effect^{ext}_{ind}(\state)$ for each $i \in [n]$ \cite{KiNuIm10}. Given this decomposition of the effect $\oA_i$ we have that
\begin{align}\label{eq:norm-A_i}
    \no{\oA_i} = \no{\sum_{j=1}^{r_i} \alpha^{(i)}_j e^{(i)}_j } \leq  \sum_{j=1}^{r_i} \alpha^{(i)}_j \no{e^{(i)}_j } = \sum_{j=1}^{r_i} \alpha^{(i)}_j
\end{align}
for all $i \in [n]$ so that in particular 
\begin{align}\label{eq:lmax-A_i-1}
    \lmax{\oA} = \sum_{i =1}^n \no{\oA_i} \leq \sum_{i =1}^n \sum_{j=1}^{r_i} \alpha^{(i)}_j.
\end{align}
On the other hand, we see that
\begin{align}\label{eq:lmax-A_i-2}
    \frac{1}{\lambda_0} =  \frac{1}{\lambda_0} u(s_0) =  \frac{1}{\lambda_0} \sum_{i =1}^n \oA_i(s_0) = \frac{1}{\lambda_0} \sum_{i =1}^n \sum_{j=1}^{r_i} \alpha^{(i)}_j e^{(i)}_j(s_0) = \sum_{i =1}^n \sum_{j=1}^{r_i} \alpha^{(i)}_j
\end{align}

If now $\lmax{\oA} = \sum_{i=1}^n \no{\oA_i} = 1/\lambda_0$, then it follows from Eq. \eqref{eq:lmax-A_i-1} and Eq. \eqref{eq:lmax-A_i-2} that
\begin{align*}
    \sum_{i =1}^n \left[ \sum_{j=1}^{r_i} \alpha^{(i)}_j - \no{\oA_i}\right] = 0.
\end{align*}
On the other hand from Eq. \eqref{eq:norm-A_i} we see that we must actually have that $\no{\oA_i} = \sum_{j=1}^{r_i} \alpha^{(i)}_j$ for all $i \in [n]$. Since the maximal value of each $e^{(i)}_j \in \effect^{ext}_{ind}(\state)$ is attained by a unique pure state $s_{e^{(i)}_j} \in \state^{ext}$ such that $s_{e^{(i)}_j} \neq s_{e^{(i)}_k}$ for all $j \neq k$ for all $i \in [n]$, it follows that $r_i=1$ for all $i \in [n]$ so that $\oA$ is indecomposable.
\end{proof}

\begin{example}[Quantum theory]\label{ex:is-quantum}
In $d$-dimensional quantum theory $\state^q_d$ with a $d$-dimensioanl Hilbert space $\hi$  we can take $s_0 = \frac{1}{d} \id$ so that then $\tr{s_0 \kb{\varphi}{\varphi}}= 1/d$ for all unit vectors $\varphi \in \hi$. Thus, condition i) in Prop. \ref{prop:inf-stor-condition} is satisfied and then $\is{\state^q_d} = d$ and it is attained by all indecomposable measurements, i.e., all POVMs with rank-1 effects. On the other hand, since the supremum in  $\sup_{\varrho \in \state^q_d} \tr{\varrho \kb{\varphi}{\varphi}}$ attained by a unique state for each unit vector $\varphi \in \hi$, also the condition ii) in Prop. \ref{prop:inf-stor-condition} is satisfied and thus $\is{\state^q_d} = d$ is attained only by rank-1 POVMs.
\end{example}

\begin{example}[Classical theory]
Let's explore the conventional representation of a classical system in phase space \cite{PSAQT82}. In this context, each dimension within the phase space corresponds to a degree of freedom of the system. Consequently, the points in the phase space uniquely determine the state of the system. When viewed as a statistical (operational) theory, this representation necessitates an extension of the concept of states to encompass all probability distributions across the phase space. The concept of a \emph{simplex} precisely encapsulates this extension. If $\Omega$ denotes a (finite) phase space with $d+1$ points, the set of probability distributions on $\Omega$ forms a $d$-simplex. In other words, it constitutes the convex hull of its $d+1$ affinely independent extreme points. Therefore, a theory is deemed classical if and only if its state space is a simplex.

We denote the state space of a $d$-dimensional classical theory by $\state^{cl}_d$ and we may represent it as a $d-1$-simplex as
\begin{equation}
    \state^{cl}_d = \left\lbrace p=(p_1, \ldots, p_d) \in \real^d \, | \, \forall i: \ p_i \geq 0, \ \sum_{i=1}^d p_i =1 \right\rbrace .
\end{equation}
We note that the pure states in this case are the deterministic probability distributions, i.e., states $p$ which have $p_i=1$ for some $i \in [d]$, and that every mixed state can be written as a unique convex composition of pure states. We denote the pure states by $x_i$ where $i \in [d]$. The effect space $\effect(\state^{cl}_d)$ is then generated by linear functions $b_j: \real^d \to [0,1]$ defined as $b_j(x_i) = \delta_{ij}$ for all $i,j \in [d]$, where $\delta$ is the Kronecker delta. More precisely, any effect $e \in \effect(\state^{cl}_d)$ can be written as $e = \sum_j \nu_{j} b_j$ for some coefficients $\{\nu_{j}\}_{j} \subset [0,1]$. In particular, we have the unit effect $u = \sum_j b_j$. The extreme effects are of the form $\sum_{j \in J} b_j$ for some $J \subset [d]$ and the indecomposable effects are all multiples of some $b_j$. It follows that $\effect^{ext}_{ind}(\state^{cl}_d) = \{b_j\}_{j=1}^d$.

Let us denote $s_0 = \frac{1}{d} \sum_{i=1}^{d} x_i$. It follows that for all extreme indecomposable effects $b_j$ we have that $b_j(s_0) = 1/d$ for all $j \in [d]$. Thus, the condition i) in Prop. \ref{prop:inf-stor-condition} is satisfied and thus $\is{\state^{cl}_d} = d$ and it is attained by all indeocomposable measurements, i.e., measurements with effects proportional to some $b_j$. On the other hand, for each $j \in [d]$ clearly the supremum in $\sup_{s \in \state^{cl}_d} b_j(s)$ is attained by a unique state $x_j$ so that also the condition ii) in Prop. \ref{prop:inf-stor-condition} is satisfied and thus $\is{\state^{cl}_d} = d$ is attained only by indecomposable measurements.
\end{example}

\begin{example}[Polygon theories]\label{ex:is-polygon}

In alignment with \cite{JaGoBaBr11}, we define a state space denoted as $\poly_n$ for a regular $n$-sided polygon, embedded in $\real^3$. This state space is defined as the convex hull of its $n$ extreme points, represented by:
\begin{equation}
s_k=
\begin{pmatrix}
r_n \cos\left(\dfrac{2 k \pi }{n}\right) \\
r_n \sin\left(\dfrac{2 k \pi }{n}\right) \\
1
\end{pmatrix}, \quad k \in [n]. 
\end{equation}
Here, $r_n = \sqrt{\sec\left( \frac{\pi}{n}\right)}$. Within this framework, the effect space encompasses the zero effect, denoted as $o = (0, 0, 0)^T$, and the unit effect, denoted as $u = (0, 0, 1)^T$. Let us define $s_0 = (0,0,1)^T$. The symmetry properties of the state space depend on the parity of $n$; in particular the state space may or may not exhibit point-symmetry around $s_0$. This gives rise to distinct structural characteristics in the effect space $\effect(\poly_n)$ for odd and even values of $n$.

For even $n$, the non-trivial extreme points are expressed as:
\begin{equation}
e_k= \dfrac{1}{2}
\begin{pmatrix}
r_n \cos\left(\dfrac{(2k-1) \pi }{n}\right) \\
r_n \sin\left(\dfrac{(2k-1) \pi }{n}\right) \\
1
\end{pmatrix}, \quad k \in [n],
\end{equation}
resulting in $\effect(\poly_n) = \conv{\left(\{o,u,e_1, \ldots,e_n\}\right)}$. Notably, all non-trivial extreme effects lie on a common hyperplane determined by the effects $e$ for which $e(s_0) = 1/2$. The indecomposable effects are expressed as multiples of $e_k$ for some $k \in [n]$.

In the case of odd $n$, the effect space exhibits $2n$ non-trivial extreme effects:
\begin{equation}
g_k
= \dfrac{1}{1+r^2_n}
\begin{pmatrix}
r_n \cos\left(\dfrac{2k \pi }{n}\right) \\
r_n \sin\left(\dfrac{2k \pi }{n}\right) \\
1
\end{pmatrix}, \quad \quad f_k = u-g_k
\end{equation}
for $k \in [n]$. In this scenario, $\effect(\poly_n) = \conv{\left(\{o,u, g_1, \ldots, g_n, f_1 ,\ldots, f_n\}\right)}$, and the non-trivial extreme effects are distributed across two distinct planes determined by all points $g$ and $f$ such that $g(s_0) = \frac{1}{1+r^2_n}$ and $f(s_0) = \frac{r^2_n}{1+r^2_n}$. The indecomposable effects are expressed as multiples of $g_k$ for some $k \in [n]$. It is noteworthy that $\poly_3$ is isomorphic to the classical state space $\state^{cl}_3$ as it corresponds to a triangle, i.e., a $2$-simplex.

Since $\effect^{ext}_{ind}(\poly_n) = \{e_k\}_{k=1}^n$ for even $n$ and $\effect^{ext}_{ind}(\poly_n) = \{g_k\}_{k=1}^n$ for odd $n$ and $e_k(s_0) =1/2$ and $g_k(s_0)=1+r^2_n$ for all $k \in [n]$, we see by condition i) in Prop. \ref{prop:inf-stor-condition} that $\is{\poly_n}= 2$ for even $n$ and $\is{\poly_n}=1+r^2_n=1+\sec(\pi/n)$ for odd $n$ and that these values are attained for all indecomposable measurements. Here we note that for odd $n$ we have that $\is{\poly_n} > 2 = \dop_{\poly_n}$ so that the odd polygon theories manifest super information storability. For pentagon state space this was first recognized in \cite{MaKi18} and for any odd $n$ this was shown in \cite{HeLe22}. 

Furthermore, for odd $n$ one can readily confirm that the supremum in $\sup_{i \in [n]} g_k(s_i)$ is attained by a unique state $s_k$ for all $k \in [n]$ so that by the condition ii) in Prop. \ref{prop:inf-stor-condition} the value $\is{\poly_n}=1+\sec(\pi/n)$ is attained only by indecomposable measurements. On the other hand one can easily check that in the even case $e_k(s_k)=e_k(s_{k-1})=1$ so that for even $n$ condition ii) in Prop. \ref{prop:inf-stor-condition} is not satisfied. For example, in square state space ($n=4$) we can define a decomposable measurement $\oA$ by setting $\oA_1 =1/2(e_1 + e_2)$ and $\oA_2 = 1/2(e_3+e_4)$ for which now clearly $\lmax{\oA}=2=\is{\poly_4}$.
\end{example}

%%%%%%%%%%%%%%%%%%%%%%%%%%%%%
\section{Optimal strategy in the information-storing game}
%%%%%%%%%%%%%%%%%%%%%%%%%%%%

Let us now focus on solving the information storing game that was presented previously. As a reminder: the game consists of you choosing $n\in \nat$, states $\sX= \{s_i\}_{i=1}^n$ and an $n$-outcome measurement $\oA \in \obs_n(\state)$ such that in the case of you successfully identifying a randomly given state from $\sX$ by using $\oA$ you are awarded $n$ utiles and in the case when you get the state wrong you are given $w \leq 0$ utiles (which is fixed by the game). As presented before in Eq. \eqref{eq:reward}, then the average reward of the information storing game is given by
\begin{equation}\label{eq:expected-reward}
   \expec_w(\sX,\oA) = w + \left( 1- \frac{w}{n} \right) \sum_{i=1}^n \oA_i\left(s_i\right) \, .
\end{equation}
As a player one wants to maximize this reward function by choosing the best $n$, $\sX$ and $\oA$.

\subsection{The case with no penalty}

Let us first solve the simplest case when $w=0$ so that there is no penalty. In this case the game reduces to finding the optimal strategy for maximizing the product of the number of states that is used to encode a message and the optimal success probability of decoding the message by performing a minimum-error discrimination task on the states. In other words, we are looking for a strategy that can be used to attain the information storability $\is{\state}$ of the theory:
\begin{equation}
    \expec_0 := \sup_{n \in \nat} \sup_{\oA \in \obs_n(\state)} \sup_{\sX \subset \state, |\sX|=n}  \expec_0(\sX,\oA) = \is{\state} \, .
\end{equation}
Thus, the optimal value is attained for all $n \in \nat$ such that $\isn{\state} = \is{\state}$. We have earlier introduced $\nsat_{\state}$ as the smallest integer $n$ such that $\isn{\state}=\is{\state}$. Furthermore, since $\isn{\state} = \is{\state}$ for all $n \geq \nsat_{\state}$, it does not make a difference for the expected reward if we choose more states than $ \nsat_{\state}$ as long as they can be used to attain $\is{\state}$. We can present this result as follows:

\begin{proposition}[No penalty]
The maximal average expected reward $\expec_w$ in the information storing game in a theory in the case of no penalty, $w=0$, is 
\begin{equation}
    \expec_0 = \is{\state} \, .
\end{equation}
This is attained with $n \geq \nsat_{\state}$ maximally decodable states and with an $n$-outcome measurement with maximal decoding power.
\end{proposition}

As was pointed out before, it is known that there are state spaces $\state$ where $\is{\state} > \dop_{\state}$. Thus, \emph{the information storing game without penalty can be seen as an operational test for super information storability}. 

\subsection{The case with penalty}

Let now $w <0$ in Eq. \eqref{eq:expected-reward}. As was already pointed out in Sec. \ref{sec:encoding-decoding} the maximizing set of states $\sX$ and the maximizing measurement $\oA$ should now be nondegradable: indeed, if $\sX$ would be degradable, then there would exist a subset $\sY \subset \sX$ with $m<n$ states which would preserve the value of $\mu_{\max}$ but nevertheless increase the expected reward since then $-\frac{w}{m} > -\frac{w}{n} > 0$. On the other hand, if the maximizing measurement $\oA$ with $n$ outcomes was degradable, then there would be another measurement $\oB$ with $k < n$ outcomes that would preserve the value of $\lambda_{\max}$ but nevertheless increase the expected reward since then again $-\frac{w}{k} > -\frac{w}{n} > 0$. Thus, we have shown the following result:

\begin{proposition}\label{prop:optimal-no-penalty}
The optimal strategy in the information-storing game is attained by using a nondegradable measurement and nondegradable set of states.
\end{proposition}

Let us now continue with the optimization. It is clear that when we perform the maximization over $\sX$ and $\oA$ where they are assumed to have $n$ states and $n$ outcomes respectively, then we get
\begin{align}
    \expec_w(n) &:= \sup_{\oA \in \obs_n(\state)} \sup_{\sX \subset \state, |\sX|=n}  \expec_w(\sX,\oA) \nonumber \\
    &= w + \left( 1- \frac{w}{n} \right)  \sup_{\oA \in \obs_n(\state)} \sup_{\sX \subset \state, |\sX|=n} \sum_{i=1}^n \oA_i\left(s_i\right) \nonumber \\
    &=  w + \left( 1- \frac{w}{n} \right)  \sup_{\oA \in \obs_n(\state)} \lmax{\oA} \nonumber \\
    &=  w + \left( 1- \frac{w}{n} \right) \isn{\state} \nonumber \\
    &= w + \isn{\state} -w  \frac{\isn{\state}}{n} \, . \label{eq:optimal-n}
\end{align}
What remains to be chosen is $n$ that maximizes the above expression to get the maximal average expected reward $\expec_w := \sup_{n \in \nat} \expec_w(n)$. We note that both the terms $\isn{\state}$ and $-w  \frac{\isn{\state}}{n}$ are now positive so that we need to maximize both $\isn{\state}$ and $\frac{\isn{\state}}{n}$ depending on the value of $w$.

\subsubsection{Theories without super information storability}

Let first $\state$ be a state space which does not have super information storability so that $\is{\state} = \dop_{\state}$. In this case we see that for all $n \leq \dop_{\state}$ we have that $\isn{\state} = n$ and thus $\frac{\isn{\state}}{n}= 1$. On the other hand for all $n > \dop_{\state}$ we have that $\isn{\state}= \dop_{\state}$ and $\frac{\isn{\state}}{n} = \frac{\dop_{\state}}{n}< 1$. Hence, the optimal value for $n$ in this case is $\dop_{\state}$ since it maximizes both  $\isn{\state}$ and $\frac{\isn{\state}}{n}$ irrespective of $w$. Thus, we have the following:

\begin{proposition}[No super information storability]
The maximal average expected reward $\expec_w$ in the information storing game in a theory without super information storability is 
\begin{equation}
    \expec_w = \dop_{\state} 
\end{equation}
irrespective of the exact value of the penalty $w < 0$. This is attained only with $\dop_{\state}$ perfectly distinguishable states and with a measurement that distinguishes them.
\end{proposition}

\subsubsection{Theories with super information storability}

Let us now turn to theories which do have super information storability. Thus, let now $\state$ be a state space such that $\is{\state} > \dop_{\state}$. 

First, let us note that one can always use a strategy which results in expected reward of $\dop_{\state}$ utiles by using $\dop_{\state}$ perfectly distinguishable states and a measurement that distinguishes them: namely, then $\isn{\state} = \dop_{\state}$ and $\frac{\isn{\state}}{n}= 1$ for $n = \dop_{\state}$ so that from Eq. \eqref{eq:optimal-n} we see that 
\begin{equation} \label{eq:adv-over-perf-discr}
    \expec_w \geq \expec_w(\dop_{\state}) = \dop_{\state} \, .
\end{equation}

Let us now look for strategies that would result in a higher reward. Thus, we will look for cases when $\expec_w(n) > \dop_{\state}$. From Eq. \eqref{eq:optimal-n} we see that this can happen only when 
\begin{equation}
    w \left( 1-\frac{\isn{\state}}{n}\right) > \dop_{\state} - \isn{\state} \, .
\end{equation}

Now if we have $n \leq \dop_{\state}$, then the LHS of the above equation is zero since $\isn{\state} = n$ so that the above equation is satisfied only when $\isn{\state} > \dop_{\state}$ which is impossible for $n \leq \dop_{\state}$. Thus, we can focus on looking at cases when $n > \dop_{\state}$. For $n > \dop_{\state}$ we have that $\isn{\state} < n$ so that the LHS of the above expression is strictly negative. Then it can be satisfied only if $\dop_{\state} - \isn{\state} < 0$, i.e., only when we are witnessing super information storability. Thus, if we denote by $\mdet_{\state}$ the minimum number $m \in \nat$ such that super information storability is detected, i.e., when $\ism{\state} > \dop_{\state}$, we see that $\dop_{\state} - \isn{\state} < 0$ only for all $n \geq \mdet_{\state}$. Thus, we may further restrict to look only for $n \geq \mdet_{\state}$.

Solving for $w$ in Eq. \eqref{eq:adv-over-perf-discr} we get the following result:
\begin{proposition}
     In the information storing game with penalty $w<0$ an advantage over the perfect discrimination strategy can be attained with $n$ states $\sX$ and an $n$-outcome measurement $\oA$ if and only if 
     \begin{equation}
        w > n\frac{\dop_{\state}- \isn{\state}}{n- \isn{\state}} \quad \mathit{and} \quad n \geq \mdet_{\state} \, .
    \end{equation}
    Furthermore, then we can and must choose $\sX$ and $\oA$ such that $\mmax{\sX} = \lmax{\oA | \sX} = \isn{\state}$ in order to see the advantage.
\end{proposition}

%\blue{Useful bounds for $w$?}

To see more specifically for which values of $n$ it is possible, let us try to solve for $n$ in Eq. \eqref{eq:adv-over-perf-discr}. By rearranging some terms we find that it is possible only when:
\begin{equation}
    n \cdot \frac{\isn{\state}-\dop_{\state}+w}{w} < \isn{\state} \, .
\end{equation}
We see that there will be different cases wether $ \isn{\state} - \dop_{\state} +w$ is positive or negative. Simply solving for $n$ in these cases leads to the next result:

\begin{proposition}
    In the information storing game with penalty $w<0$ an advantage over the perfect discrimination strategy can be attained with $n$ states $\sX$ and an $n$-outcome measurement $\oA$ if and only if 
    \begin{equation}
        \mdet_{\state} \leq n
    \end{equation} 
    when $ \isn{\state} - \dop_{\state} \geq -w$, and 
    \begin{equation}
     \mdet_{\state} \leq n < \frac{\isn{\state} w}{\isn{\state}-\dop_{\state}+w}
    \end{equation} 
    when $ \isn{\state} - \dop_{\state} < -w$.
    Furthermore, then we can and must choose $\sX$ and $\oA$ such that $\mmax{\sX} = \lmax{\oA | \sX} = \isn{\state}$ in order to see the advantage.
\end{proposition}

%\blue{Useful bounds for $n$?}

The previous two results can also be used, if possible, to find a penalty for which the optimal strategy is something other than perfect discrimination strategy. We illustrate this in Figs. \ref{fig:P5xS2} and \ref{fig:P5+P7} for different combinations of odd polygon state space where we know that super information storability is manifested (see Example \ref{ex:is-polygon}).

\begin{figure}
\centering
\includegraphics[width=0.8\textwidth]{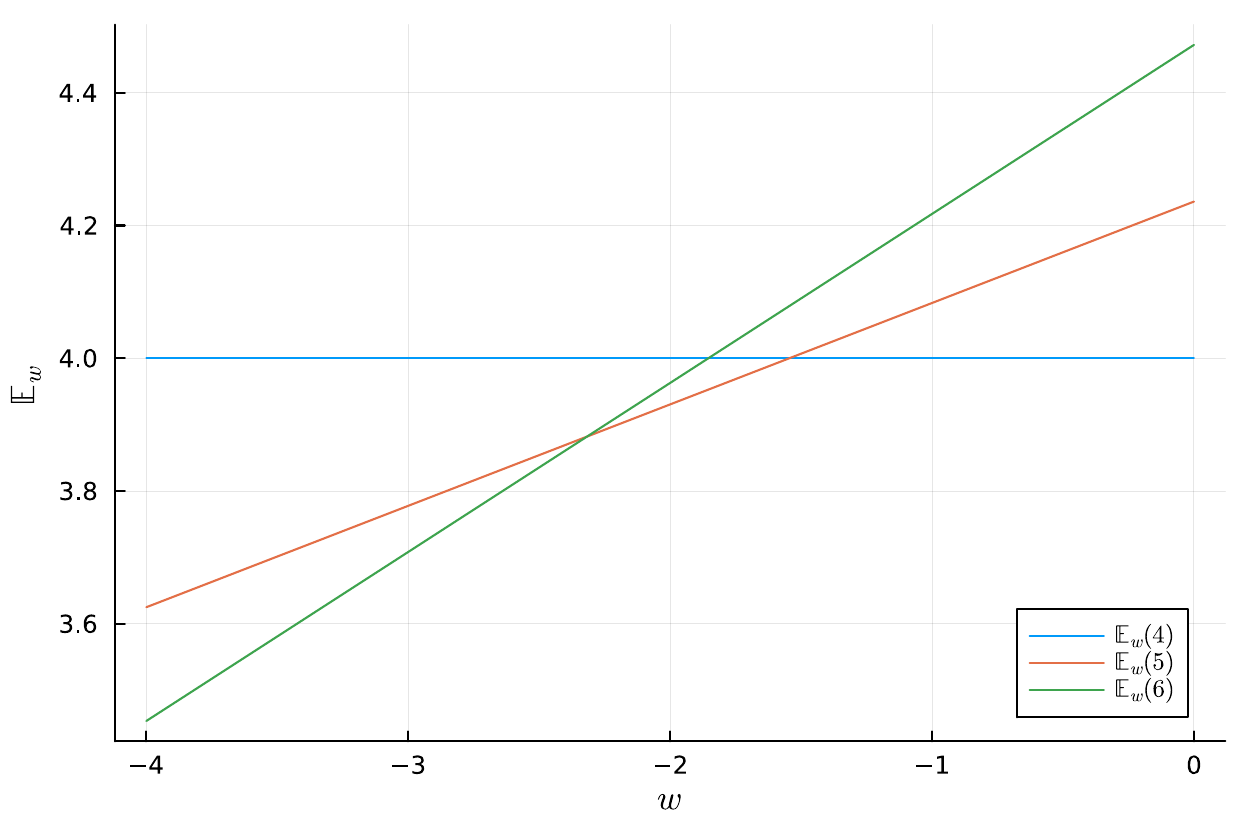}
\caption{The expected reward $\expec_w(n)$ as a function of the penalty $w$ for the composite state space $\state= \poly_5 \otimes \state^{cl}_2$  for $\dop_{\state} \leq n \leq \nsat_{\state}$. For this state space we have $\dop_{\state} = 4$ and $\nsat_{\state} = 6$ so it manifests super information storability.}
\label{fig:P5xS2}
\end{figure}

\begin{figure}
\centering
\includegraphics[width=0.8\textwidth]{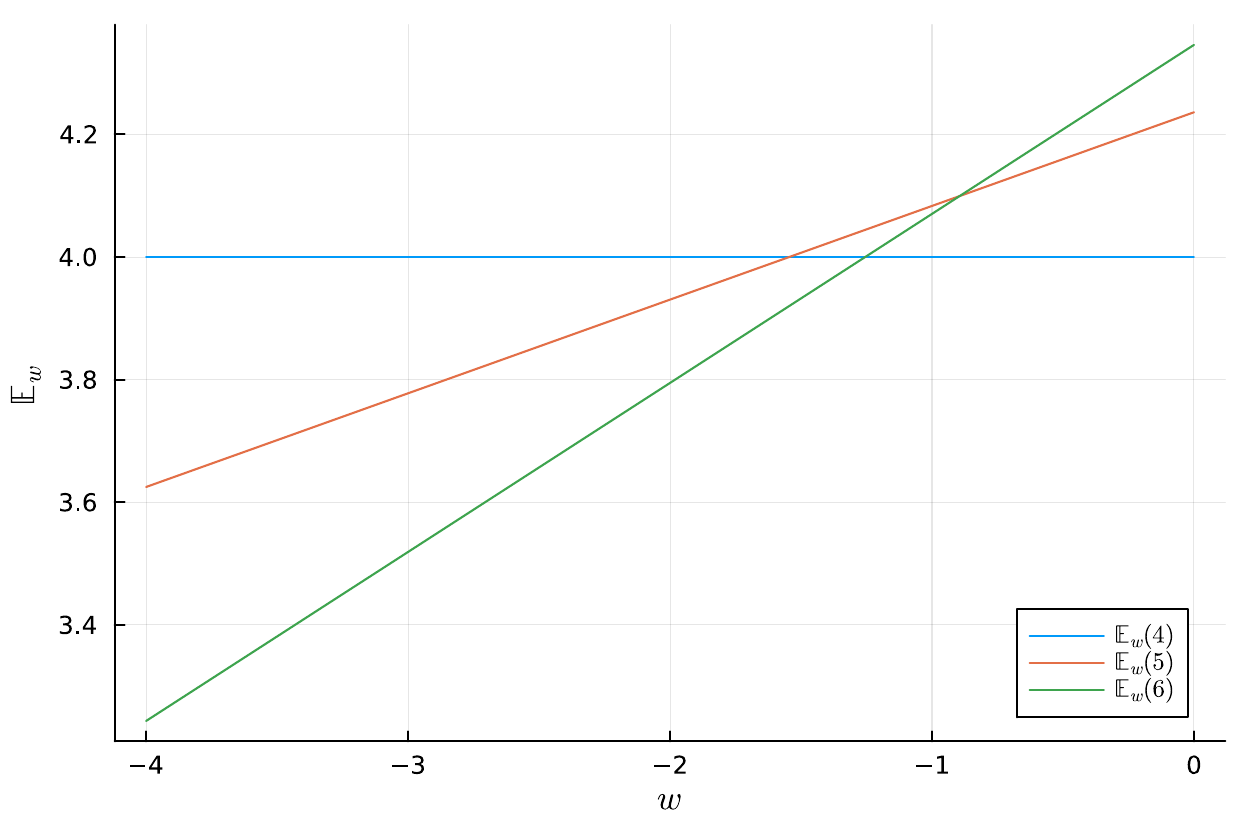}
\caption{The expected reward $\expec_w(n)$ as a function of the penalty $w$ for the direct sum state space $\state= \poly_5 \oplus \poly_7$ for $\dop_{\state} \leq n \leq \nsat_{\state}$. Also for this state space we have $\dop_{\state} = 4$ and $\nsat_{\state} = 6$ so it manifests super information storability.}
\label{fig:P5+P7}
\end{figure}

 Above we have now considered if there is some better strategy than the perfect discrimination strategy. How about the optimal strategy? To this end let us start by comparing these super information storability strategies with different $n$. As explained above, the perfect discrimination strategy is optimal for all $n < \mdet_{\state}$. Also, it is clear that if $n \geq \nsat_{\state}$, then $\expec_w(n)$ monotonically decreases (see Eq. \eqref{eq:optimal-n}). What happens between for the expected reward between $\mdet_{\state} \leq n < \nsat_{\state}$? Is it monotonically increasing in all theories for all penalties or not so that the optimal strategy will always be achieved with $n = \nsat_{\state}$? We can show that the answer to the previous question is negative: While in Fig. \ref{fig:P5xS2} for $\state= \poly_5 \otimes \state^{cl}_2$ (a tensor product of pentagon and classical bit) the optimal solution jumps directly from the perfect distinguishability strategy to the maximal super information storability strategy with $n = \nsat_{\state}$ depending on the penalty, in Fig. \ref{fig:P5+P7} we see that for the state space $\state= \poly_5 \oplus \poly_7$ (direct sum of pentagon and heptagon, see \cite{HeLePl19} for the definition of a direct sum) there are choices of penalties for which the optimal strategy is a non-maximal super information storability strategy with $\mdet_{\state} \leq n < \nsat_{\state}$. This shows that the behaviour of the expected reward $\expec_w(n)$ in the interval $\mdet_{\state} \leq n < \nsat_{\state}$ depends on the theory and thus the optimal strategy may not be as straightforward to solve in general. We leave the analytical optimal solution as an open problem.

%\begin{align}
%    \expec_w(n) > \expec_w(m) \quad \Leftrightarrow \quad w \left( \frac{\ism{\state}}{m}-\frac{\isn{\state}}{n} \right) > \ism{\state} - \isn{\state}
%\end{align}
%If $m>n$, then RHS is positive or zero, but how about LHS?

\section{Conclusions} 
In this work we have introduced the party game of the year, namely the information-storability game: your task is to freely choose a set of states from which a state is chosen at random and by measuring the state you must guess which state was chosen. If the guess is correct you receive as many utiles as the number of states in the set that you choose and if you guess incorrectly you lose a fixed number of utiles. By formalizing the game we see that two key concepts in the solution are the notions of the information storability of the theory and nondegradability of the used set of states and the used measurement. We have developed the theory of both of these concepts by giving conditions on when the set of states or the measurement is degradable, and how information storability can be calculated in specific theories.

By using these concepts we have explored the optimal solutions in three different cases: i) game with no penalty, where the optimal strategy is always attained by using the strategy that maximizes the information storability and depending if the theory has super information storability or not this optimal strategy may or may not be the perfect discrimination strategy, ii) game with penalty in a theory without super information storability, where the optimal strategy is just the maximal perfect discrimination strategy, and iii) game with penalty in a theory with super information storability, where the optimal strategy will be either some super information storability strategy or the perfect distingishability strategy depending on the value of the penalty. What follows from our analysis is one of our most important observations that the game can be used as an operational test for super information storability.

Interestingly, in the last case iii) we have demonstrated that depending on the theory and the penalty neither the perfect discrimination strategy or the maximal super information storability strategy may not be the optimal one but some other super information strategy may perform better. We leave it as an open question to solve the optimal strategy analytically in this case but for now the lack of perfect knowledge of the solution adds to the fun of the game even for the creators!

%\newpage

\section*{Acknowledgments}
TH acknowledges support from the Business Finland under the projects TORQS, Grant 8582/31/2022, BEQAHcreation, Grant 6956/31/2023, and from the Academy of Finland under the project DEQSE, Grant 349945.

LL acknowledges support from the European Union’s Horizon 2020 Research and Innovation Programme under the Programme SASPRO 2 COFUND Marie Sklodowska-Curie grant agreement No. 945478 as well as from projects APVV-22-0570 (DeQHOST) and VEGA 2/0183/21 (DESCOM).

MP acknowledges support from the Deutsche Forschungsgemeinschaft (DFG, German Research Foundation, project numbers 447948357 and 440958198), the Sino-German Center for Research Promotion (Project M-0294), the German Ministry of Education and Research (Project QuKuK, BMBF Grant No. 16KIS1618K), the DAAD, and the Alexander von Humboldt Foundation.

\bibliographystyle{unsrt}
%\bibliographystyle{alpha}
%\bibliographystyle{abbrv}
%\bibliographystyle{acm}
%\bibliographystyle{Nabbrv}
%\bibliographystyle{amsplain}
%\bibliographystyle{phcpc}
%\bibliography{bibliography}

\begin{thebibliography}{10}

\bibitem{PSAQT82}
A.S. Holevo.
\newblock {\em Probabilistic and Statistical Aspects of Quantum Theory}.
\newblock North-Holland Publishing Co., Amsterdam, 1982.

\bibitem{SSQT01}
A.S. Holevo.
\newblock {\em Statistical Structure of Quantum Theory}.
\newblock Springer-Verlag, Berlin, 2001.

\bibitem{QSCI12}
A.~S. Holevo.
\newblock {\em Quantum systems, Channels, Information: a Mathematical
  Introduction}.
\newblock Walter de Gruyter GmbH, Berlin, 2012.

\bibitem{Holevo73b}
A.S. Holevo.
\newblock Bounds for the quantity of information transmitted by a quantum
  communication channel.
\newblock {\em Probl. Peredachi Inf.}, 9:3--11, 1973.

\bibitem{MaKi18}
K.~Matsumoto and G.~Kimura.
\newblock Information storing yields a point-asymmetry of state space in
  general probabilistic theories.
\newblock arXiv:1802.01162 [quant-ph], 2018.

\bibitem{LamiThesis}
L.~Lami.
\newblock Non-classical correlations in quantum mechanics and beyond.
\newblock {\em PhD thesis. arXiv preprint arXiv:1803.02902}, 2018.

\bibitem{Plavala23}
M.~Plávala.
\newblock General probabilistic theories: An introduction.
\newblock {\em Phys. Rep.}, 1033:1--64, 2023.

\bibitem{LeppajarviThesis}
L.~Lepp\"{a}j\"{a}rvi.
\newblock Measurement simulability and incompatibility in quantum theory and
  other operational theories.
\newblock {\em PhD thesis. arXiv preprint arXiv:2106.03588}, 2021.

\bibitem{KiNuIm10}
G.~Kimura, K.~Nuida, and H.~Imai.
\newblock Distinguishability measures and entropies for general probabilistic
  theories.
\newblock {\em Rep. Math. Phys.}, 66:175--206, 2010.

\bibitem{QPSI10}
B.~Schumacher and M.~Westmoreland.
\newblock {\em {Q}uantum {P}rocesses, {S}ystems, and {I}nformation}.
\newblock Cambridge University Press, 2010.

\bibitem{oszmaniec2017simulating}
M.~Oszmaniec, L.~Guerini, P.~Wittek, and A.~Ac\'{\i}n.
\newblock Simulating positive-operator-valued measures with projective
  measurements.
\newblock {\em Phys. Rev. Lett.}, 119:190501, 2017.

\bibitem{guerini2017operational}
L.~Guerini, J.~Bavaresco, M.~T. Cunha, and A.~Ac\'{\i}n.
\newblock Operational framework for quantum measurement simulability.
\newblock {\em J. Math. Phys.}, 58:092102, 2017.

\bibitem{FiHeLe18}
S.N. Filippov, T.~Heinosaari, and L.~Lepp\"aj\"arvi.
\newblock Simulability of observables in general probabilistic theories.
\newblock {\em Phys. Rev. A}, 97:062102, 2018.

\bibitem{LaPaWi18}
L.~Lami, C.~Palazuelos, and A.~Winter.
\newblock Ultimate data hiding in quantum mechanics and beyond.
\newblock {\em Commun. Math. Phys.}, 361:661--708, 2018.

\bibitem{BlJeNe22}
A.~Bluhm, A.~Jen{\v{c}}ov{\'a}, and I.~Nechita.
\newblock Incompatibility in general probabilistic theories, generalized
  spectrahedra, and tensor norms.
\newblock {\em Commun. Math. Phys.}, 393:1125--1198, 2022.

\bibitem{BaKuMoHi97}
M.~Ban, K.~Kurokawa, R.~Momose, and O.~Hirota.
\newblock Optimum measurements for discrimination among symmetric quantum
  states and parameter estimation.
\newblock {\em Int. J. Theor. Phys.}, 36:1269--1288, 1997.

\bibitem{HeLe22}
T.~Heinosaari and L.~Leppäjärvi.
\newblock Random access test as an identifier of nonclassicality*.
\newblock {\em J. Phys. A: Math. Theor.}, 55:174003, 2022.

\bibitem{JaGoBaBr11}
P.~Janotta, C.~Gogolin, J.~Barrett, and N.~Brunner.
\newblock Limits on nonlocal correlations from the structure of the local state
  space.
\newblock {\em New J. Phys.}, 13:063024, 2011.

\bibitem{HeLePl19}
T.~Heinosaari, L.~Lepp\"aj\"arvi, and M.~Pl\'avala.
\newblock No-free-information principle in general probabilistic theories.
\newblock {\em Quantum}, 3:157, 2019.

\end{thebibliography}

\appendix
%%%%%%%%%%%%%%%%%%%%%%%%%%%%%%%%%%%
\section{About the definitions of degradability}\label{appendix-degradability}
%%%%%%%%%%%%%%%%%%%%%%%%%%%%%%%%%%
In Sec. \ref{sec:encoding-decoding} we introduced the concept of degradable states and measurements and linked them to the information storing game. However, although the definitions might seem intuitive in the context of the game, they do not capture the most general way of how to obtain the nondegradable states/measurements from the degradable ones. These more general processes of obtaining sets of states from other sets of states and measurements from some other measurement are called pre- and postprocessing, respectively. 

\begin{definition}
A finite set of states $\sY = \{t_j\}_{j=1}^m \subset \state$ is a \emph{preprocessing} of a finite set of states $\sX = \{s_i\}_{i=1}^n \subset \state$, denoted by $\sY \preceq \sX$, if there exists a row-stochastic matrix\footnotemark{} $C= \{C_{ji}\}_{j \in [m],i \in [n]}$ such that $t_j = \sum_{i=1}^n C_{ji} s_i$ for all $j \in [m]$.
\end{definition}

The interpretation of preprocessing follows from the stochasticity of the preprocessing matrix; namely, the preprocessd states are nothing but convex combinations of other states. Since the state space is convex it is clear that by taking any preprocessing of any set of states leads to another set of states. 

\begin{definition}
A measurement $\oB \in \obs(\state)$ with $n_{\oB}$ outcomes is a \emph{postprocessing} of a measurement $\oA \in \obs(\state)$ with $n_{\oA}$ outcomes, denoted by $\oB \preceq \oA$ if there exists a row-stochastic matrix\footnotemark[\value{footnote}] $\nu=(\nu_{ij})_{i \in [n_{\oA} ], j \in [n_{\oB} ]}$ such that $\oB_y = \sum_{i =1}^{n_{\oA}} \nu_{ij} \oA_i$ for all $j \in [n_{\oB}]$.
\end{definition}
\footnotetext{An $n \times m$-matrix $M=(M_{ij})_{i \in [n], j \in [m]}$ is row-stochastic if $M_{ij} \geq 0$ for all $i \in [n]$, $j \in [m]$ and $\sum_{j=1}^m M_{ij} =1$ for all $i \in [n]$.}

The interpretation of postprocessing is as follows: given a postprocessing matrix $\nu$ the matrix element $\nu_{ij}$ describes the transition probability of an outcome $i$ to be transformed into outcome $j$. Then the stochasticity of the postprocessing matrix guarantees that the postprocessed measurement remains a valid measurement.

Next we will see that even if we use pre- and postprocessings in the definitions of degradability they will naturally reduce to the ones presented in Sec. \ref{sec:encoding-decoding}.

\subsubsection{Degradability of states defined via preprocessing}

First we note that when $\sY  = \{t_j\}_{j=1}^m \preceq \sX = \{s_i\}_{i=1}^n$ via some preprocessing matric $C$ which takes only values $C_{ji} \in \{0,1\}$ we see that then for all $j \in [m]$ there exists an index $i_j \in [n]$ such that $t_j =s_{i_j}$. Thus, taking preprocessings of a set of states is a natural generalization of taking subsets of states. For preprocessing to be useful in the context of degradability of states we need to see how the encoding power of states behaves under preprocessing.

\begin{proposition}
Let $\sY, \sX \subset \state$ be finite subsets of states such that $\sY \preceq \sX$. Then
\begin{equation}
\mmax{\sY}\leq \mmax{\sX} \, .
\end{equation}
\end{proposition}
\begin{proof}
Let $\sY= \{t_j\}_{j=1}^m$ and $\sX = \{s_i\}_{i=1}^n$ be subsets of $m$ and $n$ states respectively. Since $\sY \preceq \sX$, there exists a stochastic matrix $C= \{C_{ji}\}_{j \in [m],i \in [n]}$ such that $t_j = \sum_{i=1}^m C_{ji} s_i$ for all $j \in [m]$. Let $\oB$ be some $m$-outcome measurement that maximizes $\mmax{\sY}$. Then
\begin{align*}
    \mmax{\sY} &= \sum_{j=1}^m \oB_j(t_j) = \sum_{j=1}^m \oB_j\left( \sum_{i =1}^n C_{ji}s_i \right) \\
    &= \sum_{i =1}^n \left( \sum_{j=1}^m C_{ji} \oB_j \right)(s_i) = \sum_{i=1}^n \tilde{\oB}_i (s_i) \\
    & \leq \max_{\oA} \sum_{i=1}^{n} \oA_i(s_i) = \mmax{\sX},
\end{align*}
where $\tilde{\oB}$ is a postprocessing of $\oB$ via the matrix $C$. 
\end{proof}

We thus see that $\mu_{\max}$ is a monotone for the preprocessing order of states, and in particular thus $\is{\state}$ can be always attained by some set of pure states. Now in the spirit of Def. \ref{def:degradable-states}, for a more general notion of degradibility we are looking for preprocessings of states that preserve the value of $\mu_{max}$.

\begin{definition}
A finite set $\sX\subset\state$ of $n$ states is \emph{pp-degradable} if there exists another finite set $\sY\subset\state$ of $m < n$ states that is a preprocessing of $\sX$, i.e., $\sY \preceq \sX$, such that 
\begin{equation}
\mmax{\sY} = \mmax{\sX} \, .
\end{equation}
Otherwise $\sX$ is \emph{pp-nondegradable}. 
\end{definition}

Now we can show that even though we are using a seemingly more general definition via preprocessings for pp-degradability, it still reduces to the original notion of degradavility.

\begin{proposition}
    A finite set $\sX \subset \state$ is pp-degradable if and only if it is degradable.
\end{proposition}
\begin{proof}
First, as noted above, having $\sY$ to be a proper subset of $\sX$ is a special case of preprocessing so that if $\sX$ is degradable, then it is also pp-degradable. Next we will show that pp-degradability of states reduces to degradability.  Let $\sX= \{s_i\}_{i=1}^n$ be a pp-degradable set of states. It turns out that we can choose the degrading set $\sY \preceq \sX$ to not be just any preprocessing of $\sX$ but a deterministic preprocessing where all the elements of the preprocessing matrix are either $0$ or $1$: Let $\sY' = \{t_j\}_{j=1}^m$ be some set of states such that  $\sY' \preceq \sX$ via some preprocessing matrix $C$ such that $\mmax{\sY'} = \mmax{\sX}$ for some $m < n$. Let $\oB$ be some $m$-outcome measurement that maximizes $\mmax{\sY'}$. We now see by using Lemma \ref{lemma:mmax-lmax-subset} that
\begin{align*}
    \mmax{\sY'} &= \sum_{j=1}^m \oB_j(t_j) = \sum_{j=1}^m \oB_j \left( \sum_{i=1}^n C_{ji}s_i \right) = \sum_{j=1}^m \sum_{i=1}^n C_{ji} \oB_j(s_i) \\
    &\leq \sum_{j=1}^m \sum_{i=1}^n C_{ji} \sup_{s \in \sX} \oB_j(s) = \sum_{j=1}^m \sup_{s \in \sX} \oB_j(s) = \lmax{\oB| \sX} \\
    &\leq \sup_{\oB'}  \lmax{\oB'| \sX} = \mmax{\sX} =  \mmax{\sY'} \,.
\end{align*}

Thus, we  have that $\sum_{j=1}^m \oB_j \left( \sum_{i=1}^n C_{ji}s_i \right)= \sum_{j=1}^m  \sup_{s \in \sX} \oB_j(s)$. Now, if we denote the convex hull of $\sX$ by $\conv(\sX)$, we have that $\sup_{s \in \sX} \oB_j(s) = \sup_{s \in \conv(\sX)} \oB_j(s)$ for all $j \in [m]$ because $\conv(\sX)$ is compact and because $\sX$ contains all the extreme points of $\conv(\sX)$. We can thus conclude that for all $j \in [m]$ we must have that $\oB_j(s_i) = \sup_{s \in \sX} \oB_j(s)$ for all $i \in [n]$ such that $C_{ji} \neq 0$. Furthermore, since $C$ is a stochastic matrix we know that for all $j \in [m]$ there exists at least one index $i_j \in [n]$ such that $C_{ji} \neq 0$. If we now take $\sY = \{s_{i_j}\}_{j=1}^m$ we see that $\sY \preceq \sX$ where the preprocessing is now deterministic but nevertheless $\mmax{\sY} = \mmax{\sX}$. From the requirement $m < n$ it follows that actually $\sY \subset \sX$.
\end{proof}

\subsubsection{Degradability of measurements defined via postprocessing}
Similarly as taking preprocessings of states can be seen as a probabilistic generalization of just taking subsets of states, taking postprocessings of a measurement probabilistically generalizes merging of the measurement outcomes. Namely, if $\oB \preceq \oA$ for some postprocessing matrix $\nu$ with entries $\nu_{ij} \in \{0,1\}$, then there exists a function $f: [n_{\oA}] \to [n_{\oB}]$ such that $\nu_{ij} = \delta_{f(i),j}$ for all $i \in [n_{\oA}]$ and $j \in [n_{\oB}]$ so that $\oB_j = \sum_{i \in f^{-1}(j)} \oA_x$ for all $j \in [n_{\oB}]$. 

In \cite{HeLe22} it was shown that $\lambda_{\max}$ is a monotone on the postprocessing order but we repeat the proof here in the more general set dependent case for completeness.

\begin{proposition}\label{prop:lmax-pp}
    Let $\oA, \oB \in \obs(\state)$ be measurements such that $\oB \preceq \oA$. Then
    \begin{equation}\label{eq:lmax-pp}
        \lmax{\oB |  \sX} \leq \lmax{\oA | \sX} 
    \end{equation}
    for all $\sX \subset \state$.
\end{proposition}
\begin{proof}
    If $\oA$ has $n_{\oA}$ outcomes and $\oB$ has $n_{\oB}$ outcomes, then $\oB \preceq \oA$ means that there exists a row-stochastic matrix $\nu = \{\nu_{ij}\}_{i \in [n_{\oA}], j \in [n_{\oB}]}$ such that $\oB_j = \sum_{i \in [n_{\oA}]} \nu_{ij} \oA_j$ for all $j \in [n_{\oB}]$. Now
    \begin{align*}
        \lmax{\oB | \sX} &= \sum_{j =1}^{n_{\oB}} \sup_{s \in \sX} \oB_j(s) = \sum_{j =1}^{n_{\oB}} \sup_{s \in \sX} \left( \sum_{i =1}^{n_{\oA}} \nu_{ij}\oA_i \right)(s) \\
        & \leq \sum_{j =1}^{n_{\oB}} \sum_{i =1}^{n_{\oA}} \nu_{ij}  \sup_{s \in \sX}  \oA_i (s) =  \sum_{i =1}^{n_{\oA}}   \sup_{s \in \sX}  \oA_x (s) \\
        & = \lmax{\oA | \sX}
    \end{align*}
    for all $\sX \subset \state$.
\end{proof}

In a similar fashion one can also show the following for mixtures of measurements \cite{HeLe22}:
\begin{proposition}\label{prop:lmax-convex}
     Let $\oA^{(1)},, \ldots,\oA^{(m)}, \oA \in \obs_{n}(\state)$ be $n$-outcome measurements such that $\oA = \sum_{k=1}^m p_k \oA^{(k)}$ for some probability distribution $(p_k)_{k=1}^m$. Then
    \begin{equation}\label{eq:lmax-mixing}
         \lmax{\oA | \sX}  \leq \sum_{k=1}^m p_k \lmax{\oA^{(k)} | \sX} \leq \max_{k \in [m]} \lmax{\oA^{(k)} | \sX}
    \end{equation}
    for all $\sX \subset \state$.
\end{proposition}
\begin{proof}
A straightforward calculation shows that
    \begin{align*}
        \lmax{\oA | \sX} &= \sum_{i=1}^n \sup_{s \in \sX} \oA_i(s) = \sum_{i=1}^n \sup_{s \in \sX} \left( \sum_{k=1}^m p_k \oA^{(k)}_i \right)(s)\\
        & \leq  \sum_{k=1}^m p_k  \sum_{i =1}^n \sup_{s \in \sX} \oA^{(k)}(s) = \sum_{k=1}^m p_k \lmax{\oA^{(k)}| \sX} \\
        & \leq \max_{k \in [m]} \lmax{\oA^{(k)}| \sX} 
    \end{align*}
    for all $\sX \subset \state$.
\end{proof}

In particular, since it is known that indecomposable measurement are exactly the maximal ones with respect to the postprocessing preorder on the set of measurements \cite[Prop. 1]{FiHeLe18}, we get by combining the two previous results that $\is{\state}$ can be always attained by some extreme indecomposable measurements (also known as simulation irreducible measurements).

Now the original definition of degradability Def. \ref{def:degradable-measurement} can be restated in terms of postprocessing as follows:

\begin{definition}
A measurement $\oA$ with $n_\A$ outcomes is \emph{pp-degradable} if there exists a measurement $\oB \preceq \oA$ with $n_\B<n_\A$ outcomes such that 
\begin{equation}
\lmax{\oA} = \lmax{\oB}
\end{equation}
If $\oA$ is not pp-degradable then it is \emph{pp-nondegradable}. 
\end{definition}

Similarly to pp-degradable states, we can now show that also in the case of measurements pp-degradablity reduces to degradability.

\begin{proposition}
    A measurement $\oA \in \obs(\state)$ is pp-degradable if and only if it is degradable.
\end{proposition}

\begin{proof}
As was pointed out above, if a measurement is degradable then it is also pp-degradable since merging of effects is a particular instance of postprocessing. Let now $\oA$ be a pp-degradable measurement with $n_{\oA}$ outcomes and let $\oB'$ be some measurement with $n_{\oB'} < n_{\oA}$ outcomes such that $\oB' \preceq \oA$ and $\lmax{\oB'} = \lmax{\oA}$. Thus, then there exists a postprocessing $\nu' = \{\nu'_{ij}\}_{i \in n_{\oA}, j \in n_{\oB'}}$ such that $\oB'_{j} =  \sum_{i \in [n_{\oA}]} \nu'_{ij} \oA_i$ for all $j \in [n_{\oB'}]$. Now $\nu'$ can be written as a convex mixture $\nu' = \sum_{k=1}^m p_k \nu^{(k)}$ of deterministic postprocessings $\nu^{(k)} =  \{\nu^{(k)}_{ij}\}_{i \in n_{\oA}, j \in n_{\oB'}}$ that satisfy $\nu^{(k)}_{ij} \in \{0,1\}$ for all $i \in [n_{\oA}]$, $j \in [n_{\oB'}]$ and $k \in [m]$ for some probability distribution $(p_k)_{k=1}^m$ for some $m \in \nat$. If we now define measurements $\oB^{(k)} \preceq \oA$ with $n_{\oB'}$ outcomes as $\oB^{(k)}_{j} =  \sum_{i=1}^{n_{\oA}} \nu^{(k)}_{ij} \oA_i$ for all $j \in [n_{\oB'}]$ and $k \in [m]$, then it holds that $\oB' = \sum_{k=1}^m p_k \oB^{(k)}$. From Propositions \ref{prop:lmax-pp} and \ref{prop:lmax-convex} it follows that 
\begin{equation*}
    \lmax{\oA} \geq \max_{k \in [m]} \lmax{\oB^{(k)}} \geq \lmax{\oB'} = \lmax{\oA}
\end{equation*}
from which we see that all the inequalities must actually be equalities. Let us choose $k_0 \in [m]$ such that $\lmax{\oB^{(k_0)} } = \max_{k \in [m]} \lmax{ \oB^{(k)} }$ and set $\oB :=\oB^{(k_0)}$. As can be seen above, then $\lmax{\oB} = \lmax{\oA}$ and since $n_{\oB} = n_{\oB'} < n_{\oA}$, this means that the deterministic postprocessing $\nu^{(k_0)}$ just represents the merging of (some of) the effects of $\oA$. Thus, this shows that $\oA$ is degradable.
\end{proof}

\section{Proof of Lemma \ref{lemma:mmax-lmax-subset}} \label{appendix:lemma-1}

\begin{lemma*}
For any finite subset of states $\sX \subset \state$ we have that 
\begin{equation}
\mmax{\sX} = \sup_{\oA} \lmax{\oA | \sX}
\end{equation}
and they both can maximized by the same measurement. 
\end{lemma*}
\begin{proof}
Let $\sX= \{s_i\}_{i=1}^n$ and let $\oA$ be some $n$-outcome measurement that maximizes $\mmax{\sX}$. 
We have
\begin{align}\label{eq:mmax-lmax-subset-1}
    \mmax{\sX} &= \sum_{i=1}^n \oA_i(s_i) \leq \sum_{i=1}^n \sup_{s\in \sX} \oA_i(s) = \lmax{\oA | \sX} \leq \sup_{\oB} \lmax{\oB | \sX} \, .
\end{align}

On the other hand, let $\oB$ be some $m$-outcome measurement that satisfies $ \lmax{\oB | \sX} = \sup_{\oB'} \lmax{\oB' | \sX}$. For each $i \in [m]$ let $t_i \in \sX$ be some state in $\sX$ such that $\oB_i(t_i) = \sup_{s \in \sX} \oB_i(s)$ and denote $\sY = \{t_i\}_{i=1}^m$ so that $\lmax{\oB|\sX} = \lmax{\oB|\sY} = \sum_{i=1}^m \oB_i(t_i)$. It may happen that $t_i = t_j$ for some $i \neq j$ meaning that some of the effects are maximized by the same state on $\sX$. For that reason, let us define an equivalence relation $\sim$ on $[m]$ by setting $i \sim j$ if and only if $t_i = t_j$. We can now define a new measurement $\tilde{\oB}$ with outcome set $[m]/\sim$ consisting of the equivalence classes $[[i]]$ of elements $i \in [m]$ such that 
\begin{equation*}
    \tilde{\oB}_{[[i]]} = \sum_{j \in [[i]]} \oB_{j}
\end{equation*}
for all $[[i]] \in [m]/\sim$. Let us now denote by $\tilde{s}_{[[i]]} \in \sY$ the maximizing state of each $\tilde{\oB}_{[[i]]}$ in $\sY$ for all $[[i]] \in [m]/\sim$ so that $\tilde{\oB}_{[[i]]}(s_{[[i]]}) = \sup_{t \in \sY} \tilde{\oB}_{[[i]]}(t)$ for all $[[i]] \in [m]/\sim$ and let $\tilde{\sX} = \{\tilde{s}_{[[i]]}\}_{[[i]] \in [m]/\sim}$. By construction now $\tilde{\sX}$ consists of distinct states in $\sY$ and thus also in $\sX$ so that then $\mmax{\tilde{\sX}} \leq \mmax{\sX}$, and furthermore $\lmax{\oB | \sX} = \lmax{\tilde{\oB} | \tilde{\sX}}$. It now follows that
\begin{align}\label{eq:mmax-lmax-subset-2}
     \lmax{\oB | \sX} &= \lmax{\tilde{\oB} | \tilde{\sX}} = \sum_{[[i]]} \tilde{\oB}_{[[i]]}(\tilde{s}_{[[i]]}) \leq \mmax{\tilde{\sX}} \leq \mmax{\sX}\, .
\end{align}

Hence, from Eqs. \eqref{eq:mmax-lmax-subset-1} and \eqref{eq:mmax-lmax-subset-2} we can now conclude that $\mmax{\sX} = \sup_{\oA} \lmax{\oA | \sX}$. Furthermore, by combining Eq. \eqref{eq:mmax-lmax-subset-1} and \eqref{eq:mmax-lmax-subset-2} we see that $\lmax{\oA | \sX} = \lmax{\oB | \sX} = \sup_{\oB'} \lmax{\oB' | \sX}$ so that in fact both $\mmax{\sX}$ and $ \sup_{\oA} \lmax{\oA | \sX}$ can be maximized with the same measurement.
\end{proof}

\end{document}